\documentclass[10pt]{article}

\pdfoutput=1
\usepackage{caption}
\usepackage[english]{babel}
\usepackage[utf8]{inputenc}
\usepackage{amsmath,amsthm,amssymb}
\usepackage[hmargin=0.81in,vmargin=0.81in]{geometry}

\usepackage{color}
\usepackage{graphicx}
\usepackage{fullpage}
\usepackage{microtype}
\usepackage{xspace}
\usepackage{mathtools,tensor,xcolor,afterpage,wrapfig,framed,protocolj,arydshln,tikz}
\usepackage[colorlinks, citecolor=blue]{hyperref}
\usepackage{tikz-cd}

\usepackage[framemethod=tikz]{mdframed}

\usepackage{enumitem}
\newtheorem{theorem}{Theorem}
\newtheorem{fact}[theorem]{Fact}
\newtheorem{lemma}[theorem]{Lemma}
\newtheorem{corollary}[theorem]{Corollary}
\newtheorem{definition}{Definition}
\theoremstyle{definition}
\newtheorem{game}{Game}

\newcommand{\qnote}[1]{{\color{red}{\bf [Qiang: #1]}}}

\newcommand{\ignore}[1]{}

\newcommand{\secp}{n}

\DeclareMathOperator{\Exp}{\mathbb{E}}

\newcommand{\resamp}{\mathsf{R}}

\newcommand\CC{\ensuremath{\mathbb C}\xspace}

\newcommand\RR{\ensuremath{\mathbb R}\xspace}

\newcommand\ZZ{\ensuremath{\mathbb Z}\xspace}

\newcommand{\DS}{\textsl{DS}}
\newcommand{\eqdef}{\stackrel{\operatorname{def}}{=}}

\newcommand{\bool}[1][\relax]{{\ensuremath{{\{0,1\}}^{#1}}}\xspace}

\DeclareMathOperator{\subv}{\mathsf{Subv}}
\DeclareMathOperator{\pred}{\mathsf{Pred}}
\DeclareMathOperator{\forwardpred}{\mathsf{Pred}^{\rightarrow}}
\DeclareMathOperator{\backwardpred}{\mathsf{Pred}^{\leftarrow}}
\DeclareMathOperator{\selfref}{\mathsf{Selfref}}
\DeclareMathOperator{\tv}{\mathsf{tv}}

\newcommand{\cP}{\ensuremath{\mathcal{P}}\xspace}

\newcommand{\cF}{\ensuremath{\mathcal{F}}\xspace}
\newcommand{\cG}{\ensuremath{\mathcal{G}}\xspace}
\newcommand{\cA}{\ensuremath{\mathcal{A}}\xspace}

\newcommand{\cS}{\ensuremath{\mathcal{S}}\xspace}

\newcommand{\cAtwo}{\ensuremath{\mathcal{A}}\xspace}
\newcommand{\cStwo}{\ensuremath{\mathcal{S}_\cAtwo}\xspace}

\newcommand{\cD}{\ensuremath{\mathcal{D}}\xspace}
\newcommand{\sD}{\ensuremath{\widehat{\cD}}\xspace}
\newcommand{\sE}{\ensuremath{\widehat{\cE}}\xspace}
\newcommand{\cE}{\ensuremath{\mathcal{E}}\xspace}

\newcommand{\cW}{\mathcal{W}}

\newcommand{\cC}{\mathcal{C}}

\newcommand{\rC}{\mathit{C}}

\newcommand{\badG}{\tilde{\cG}}

\newcommand\Verify{\mathsf{Verify}}

\newcommand{\negl}{\ensuremath{\mathsf{negl}}\xspace}
\newcommand{\poly}{\ensuremath{\mathrm{poly}}\xspace}

\newcommand{\source}{\operatorname{source}}
\newcommand{\target}{\operatorname{target}}

\definecolor{darkred}{rgb}{0.5, 0, 0}
\definecolor{darkgreen}{rgb}{0, 0.5, 0}
\definecolor{darkblue}{rgb}{0,0,0.5}
\definecolor{webred}{rgb}{0.5,0,0}
\definecolor{webblue}{rgb}{0,0,0.8}

\newcommand{\SPEC}{\textsc{spec}}
\newcommand{\IMPL}{\textsc{impl}}

\newtheorem*{theorem*}{Theorem}

\makeatletter
\renewcommand\subsubsection{\@startsection{subsubsection}{3}{\z@}%
                       {-18\p@ \@plus -4\p@ \@minus -4\p@}%
                       {0.5em \@plus 0.22em \@minus 0.1em}%
                       {\normalfont\normalsize\bfseries\boldmath}}
\makeatother
\title{Correcting Subverted Random Oracles}

\author{
Alexander Russell\thanks{University of Connecticut, acr@cse.uconn.edu} 
\and Qiang Tang\thanks{The University of Sydney, qiang.tang@sydney.edu.au} 
\and Moti Yung\thanks{Google Inc and Columbia University, moti@cs.columbia.edu} 
\and Hong-Sheng Zhou\thanks{Virginia Commonwealth University, hszhou@vcu.edu}
\and Jiadong Zhu\thanks{University of Connecticut, jiadong.zhu@uconn.edu} \thanks{Corresponding Author}
}

\begin{document}
\maketitle


\begin{abstract}
  The random oracle methodology has proven to be a powerful tool for
  designing and reasoning about cryptographic schemes.  In this paper,
  we focus on the basic problem of correcting faulty---or
  adversarially corrupted---random oracles, so that they can be
  confidently applied for such cryptographic purposes.
 
  We prove that a simple construction can transform a ``subverted''
  random oracle---which disagrees with the original one at a small
  fraction of inputs---into an object that is {\em indifferentiable}
  from a random function, even if the adversary is made aware of all
  randomness used in the transformation.
%
  Our results permit future designers of cryptographic primitives in
  typical 
  kleptographic settings (i.e., those permitting adversaries that
  subvert or replace basic cryptographic algorithms) to use random
  oracles as a trusted black box.
%

\ignore{
Employing randomness is paramount in cryptography.
The random oracle methodology, in particular, has proven to be a powerful tool 
for efficiently designing and reasoning about cryptographic schemes. 
In this paper, we focus on the basic problem of correcting faulty---or
adversarially corrupted---random oracles, so that they can be
confidently applied for such cryptographic purposes.  

We prove that a simple construction can transform a ``subverted''
random oracle into a construction that is {\em indifferentiable} from a random function.  To
analyze the construction, we develop a general re-sampling lemma,
which controls the effect of adversarial ``resampling'' on a
probability distribution; this is a tool of possible independent
interest.
Our techniques permit future designers of cryptographic primitives in
typical 
kleptographic settings (i.e., against subverting adversaries)
to use random oracles as a trusted black box.
%

\qnote{Crooked random oracle methodology.}
}
%
%

\end{abstract}


\pagenumbering{arabic}
\pagestyle{plain}


\section{Introduction}

The random oracle methodology of Bellare and
Rogaway~\cite{CCS:BelRog93} has proven to be a powerful tool for
designing and reasoning about cryptographic schemes.  It consists of
the following two steps: (i) design a scheme $\Pi$ in which all
parties (including the adversary) have oracle access to a common truly
random function, and establish the security of $\Pi$ in this favorable
setting; (ii) instantiate the random oracle in $\Pi$ with a suitable
cryptographic hash function (such as SHA256) to obtain an instantiated
scheme $\Pi'$.  The random oracle heuristic states that if the
original scheme $\Pi$ is secure, then the instantiated scheme $\Pi'$
is also secure.
While this heuristic can fail in various 
settings~\cite{STOC:CanGolHal98} the basic framework remains a
fundamental design and analytic tool. In this work we focus
on the problem of correcting faulty---or adversarially
corrupted---random oracles so that they can be confidently applied for
such cryptographic purposes.\footnote{This work expands on the preliminary conference version of this article which appeared in CRYPTO'18. Several improvements are made to the conference version in this paper. First, this paper fills some gaps in the security proof of the conference paper. Second, this paper applies more powerful analysis tools and establishes significantly tighter results. Third, the security proof in the conference version indicated how the major technical hurdles can be overcome without developing the full details. This article fills those details by game transitions from the construction to the perfect random function.}

Specifically, given a function $\tilde{h}$ drawn from a distribution
which agrees in most places with a uniform function, we would like to
produce a corrected version which appears uniform to adversaries with
a polynomially bounded number of queries. This model is partially
motivated by
the traditional study of ``program checking and self-correcting''
\cite{Blum88,STOC:BluKan89,STOC:BluLubRub90}: the goal in this theory
is to transform a program that is faulty at a small fraction of inputs
(modeling an evasive adversary) to a program that is correct at all
points with overwhelming probability. Our setting intuitively adapts
this classical theory of self-correction to the study of ``self-correcting
a probability distribution.''  Notably, the functions to be corrected
are structureless, instead of heavily structured. Despite that, the
basic procedure for correction and portions of the technical
development are analogous.

One particular motivation for correcting random oracles in a
cryptographic context arises from recent work studying design
and security in the subversion (i.e., \emph{kleptographic}) setting.
In this setting, the various components of a cryptographic scheme may
be subverted by an adversary, so long as the tampering cannot be
detected via blackbox testing. This is a challenging framework
(as highlighted by Bellare, Paterson, and Rogaway
\cite{C:BelPatRog14}), because many basic cryptographic techniques are
not directly available: in particular, the random oracle paradigm is
directly undermined. In terms of the discussion above, the random
oracle---which is eventually to be replaced with a concrete
function---is subject to adversarial subversion which complicates even
the first step of the random oracle methodology above. Our goal is
to provide a generic approach that can rigorously ``protect'' the
usage of random oracles from subversion and, essentially, establish a
``crooked'' random oracle methodology.

\ignore{
To be concrete, let us consider a simple real-world example: password
login in Linux/Unix system.  During system initialization, the root
user chooses a master password $\alpha$ and the system stores the
digest $\rho=h(\alpha)$, where $h$ is a given hash function normally
modeled as a random oracle.  During actual login, the operating system
receives input $x$ and accepts this password if $h(x) =\rho$. An
attractive feature of this practice is that it is still secure if
$\rho$ is accidentally leaked.
In the presence of kleptographic attacks, however, the module that
implements the hash function $h$ may be strategically subverted,
yielding a new function $\tilde{h}$ which destroys the security of the
scheme above: for example, the adversary may choose a relatively short
random string $z$ and define $\tilde{h}(y) = h(y)$ \emph{unless $y$
  begins with $z$}, in which case $\tilde{h}(zx) = x$. Observe that
$h$ and $\tilde{h}$ are indistinguishable by black-box testing ({\em
  as they differ only at an exponentially small number of places}); on
the other hand, the adversary can login as the system administrator
using $\rho$ and its knowledge of the backdoor $z$ (without knowing
the actual password $\alpha$, presenting $z \rho$ instead). Thus, an
attacker who has the luxury of deploying a subverted hash function may
exploit this later with dramatic effect on the system's security.
}

\subsection{Our contributions}
We first describe two concrete scenarios where hash functions are
subverted in the kleptographic setting.  We then express the security
properties by adapting the successful framework of
\emph{indifferentiability}~\cite{TCC:MauRenHol04,C:CDMP05} to our
setting with adversarial subversion. This framework provides a
satisfactory guarantee of modularity---that is, that the resulting
object can be directly employed by other constructions demanding a
random oracle. We call this new notion ``crooked'' indifferentiability
to reflect the role of adversary in the modeling; see below. (A formal
definition appears in Section~\ref{sec:model}.)

We prove that a simple construction involving only  \emph{public}
randomness can boost a ``subverted'' random oracle into a construction
that is indifferentiable from a random function. (Section
\ref{sec:construction}, \ref{sec:analysis}).
%
We expand on these contributions below.

\medskip
\noindent{\bf Consequences of kleptographic hash subversion.} We first
illustrate the security failures that can arise from use of hash
functions that are subverted at only a negligible fraction of inputs
with two concrete examples:

\medskip
\noindent
\emph{(1) Chain take-over attack on blockchain.} For simplicity,
consider a proof-of-work blockchain setting where miners compete to
find a solution $s$ to the ``puzzle''
$h(\textsf{pre}||\textsf{transactions}||s) \leq d$, where \textsf{pre}
denotes the hash of the previous block, \textsf{transactions} denotes
the set of valid transactions in the current block, and $d$ denotes
the difficulty parameter. Here $h$ is intended to be a strong hash
function. Note that 
the mining machines use a program $\tilde{h}(\cdot)$ (or a dedicated
hardware module) which could be designed by a clever adversary.  Now
if $\tilde{h}$ has been subverted so that $\tilde{h}(*||z)=0$ for a
randomly chosen $z$---and $\tilde{h}(x)=h(x)$ in all other
cases---this will be difficult to detect by prior black-box testing;
on the other hand, the adversary who created
$\tilde{h}$ has the luxury of solving the
proof of work without any effort for any challenge, and thus can
completely control the blockchain. (A fancier
subversion can tune the ``backdoor'' $z$ to other parts of the input
so that it cannot be reused by other parties; e.g.,
$\tilde{h}(w||z) = 0$ if $z=f(w)$ for a secret pseudorandom function
known to the adversary.)

\medskip
\noindent
\emph{(2) System sneak-in attack on password authentication.}
In Unix-style systems, during system initialization, the root user chooses a master password
$\alpha$ and the system stores the digest $\rho=h(\alpha)$, where $h$
is a given hash function normally modeled as a random oracle.
During login, 
the operating system receives input $x$ and accepts this password if
$h(x) =\rho$. An attractive feature of this practice is that it is still secure if $\rho$ is accidentally leaked.
In the presence of kleptographic attacks, however, the module that
implements the hash function $h$ may be strategically subverted,
yielding a new function $\tilde{h}$ which destroys the security of the
scheme above: for example, the adversary may choose a relatively short
random string $z$ and define $\tilde{h}(y) = h(y)$ \emph{unless $y$
  begins with $z$}, in which case $\tilde{h}(zx) = x$. As above, $h$
and $\tilde{h}$ are indistinguishable by black-box testing; on the
other hand, the adversary can login as the system administrator using
$\rho$ and its knowledge of the backdoor $z$ (without knowing the
actual password $\alpha$), presenting $z \rho$ instead.

Notice that, in the last two examples, if the target functions are subverted at a non-negligible portion of inputs, the subversion will be detected by a polynomial time black-box test with non-negligible probability.

\medskip
\noindent\textbf{The model of ``crooked'' indifferentiability.}
The problem of cleaning defective randomness has a long history in
computer science. Our setting requires that the transformation must be
carried out by a local rule and involve an exponentially small amount
of public randomness (in the sense that we wish to clean a defective
random function $h: \{0,1\}^n \rightarrow \{0,1\}^n$ with only a polynomial
length random string). 
The basic  framework of correcting a subverted random oracle is the following: 

First, a function $h: \{0,1\}^n \rightarrow \{0,1\}^n$ is drawn
uniformly at random. 
Then, an adversary may \emph{subvert} the
function $h$, yielding a new function $\tilde{h}$. The subverted
function $\tilde{h}(x)$ is described by an adversarially-chosen
(polynomial-time) algorithm $\tilde{H}^h(x)$, with oracle access to $h$. We insist that
$\tilde{h}(x) \neq h(x)$ only at a negligible fraction of
inputs.\footnote{%
  We remark that tampering with even a negligible fraction of inputs
  can have devastating consequences in many settings of interest:
  e.g., the blockchain and password examples above. Additionally, the
  setting of negligible subversion is precisely the desired parameter
  range for existing models of kleptographic subversion and
  security. In these models, when an oracle is non-negligibly
  defective, this can be easily detected by a watchdog using a simple sampling and
  testing regimen, see e.g., \cite{AC:RTYZ16}.} Next, the function
$\tilde{h}$ is ``publicly corrected'' to a function $\tilde{h}_R$
(defined below) that involves some public randomness $R$ selected
after $\tilde{h}$ is supplied.\footnote{We remark that in many
  settings, e.g., the model of classical self-correcting programs, we
  are permitted to sample fresh and ``private'' randomness for each
  query; in our case, we may only use a single polynomial-length
  random string for all points. Once $R$ is generated, it is made
  public and fixed, which implicitly defines our corrected function
  $\tilde{h}_R(\cdot)$. This latter requirement is necessary in our
  setting as random oracles are typically used as a public object---in particular, our attacker must have
  full knowledge of $R$.}

\ignore{
\begin{itemize}
\item A function $h: \{0,1\}^n \rightarrow \{0,1\}^n$ is drawn uniformly at random.
\item An adversary may \emph{subvert} the function $h$, yielding a new
  function $\tilde{h}$. The subverted function $\tilde{h}(x)$ is
  described by an adversarially-chosen algorithm $H^h(x)$, with oracle
  access to $h$, which determines the value of $\tilde{h}(x)$ after
  making a polynomial number of adaptive queries to $h$. We insist
  that $\tilde{h}(x) = h(x)$ with probability $1- \epsilon$ over
  random choice of $x$.
\item The function $\tilde{h}$ is ``corrected'' to a function
  $\tilde{h}_R$ by a procedure which may involve public randomness.
\end{itemize}
}
We wish to show that the resulting function (construction) is ``as
good as'' a random oracle, in the sense of indifferentiability. The
successful framework of indifferentiability asserts that a
construction $\rC^H$ (having oracle access to an ideal primitive $H$)
is indifferentiable from another ideal primitive $\cF$ if there
exists a simulator $\cS$ so that $(\rC^H,H)$ and $(\cF,\cS)$ are
indistinguishable to any distinguisher $\cD$.

To reflect our setting, an \emph{$H$-crooked-distinguisher ${\sD}$} is
introduced; the $H$-crooked-distinguisher ${\sD}$ first prepares the
subverted implementation $\tilde{H}$ (after querying $H$ first); then
a fixed amount of (public) randomness $R$ is drawn and published; the
construction $\rC$ uses only subverted implementation $\tilde{H}$ and
$R$. Now following the indifferentiability framework, we will ask for
a simulator $\cS$, such that $(\rC^{\tilde{H}^H}(\cdot,R),H)$ and
$(\cF,\cS^{\tilde{H}}(R))$ are indistinguishable to any
$H$-crooked-distinguisher $\sD$ (even one who knows $R$). A similar
security preserving theorem \cite{TCC:MauRenHol04,C:CDMP05} also holds
in our model. See Section \ref{sec:model} for details.


\medskip
\noindent{\bf The construction.} 
The construction depends on a parameter
$\ell = \poly(n)$ and public randomness $R = (r_1, \ldots, r_\ell)$,
where each $r_i$ is an independent and uniform element of
$\{0,1\}^n$. For simplicity, the construction relies on a family of
independent random oracles $h_i(x)$, for
$i \in \{0, \ldots, \ell\}$. (Of course, these can all be extracted
from a single random oracle with slightly longer inputs by defining
$\tilde{h}_i(x) = \tilde{h}(i,x)$ and treating the output of $h_i(x)$
as $n$ bits long.) Then we define
\[
\tilde{h}_R(x) = \tilde{h}_0\left(\bigoplus_{i = 1}^{\ell} \tilde{h}_i(x \oplus r_i)\right) = \tilde{h}_0\bigg(\tilde{g}_R(x)\bigg)\,.
\]
Note that 
the adversary is permitted to subvert the
function(s) $h_i$ by choosing an algorithm $H^{h_*}(x)$ so that
$\tilde{h}_i(x) = H^{h_*}(i,x)$. 
Before diving into the  analysis, let us first quickly demonstrate how some simpler constructions fail. 

\medskip
\noindent{\em Simpler constructions and their shortcomings.}
During the stage when the adversary manufactures the hash functions
$\tilde{h}_*=\{\tilde{h}_i\}_{i=0}^\ell$, the randomness
$R:=r_1,\ldots,r_\ell$ are not known to the adversary; they become
public in the second query phase. If the ``mixing'' operation is not
carefully designed, the adversary could choose inputs accordingly,
trying to ``peel off'' $R$. We discuss a few examples:
\begin{enumerate}
\item $\tilde{h}_R(x)$ is simply defined as
  $\tilde{h}_1(x\oplus r_1)$. A straightforward attack is as follows:
  the adversary can subvert $h_1$ in a way that $\tilde{h}_1(m)=0$ for
  a random input $m$; the adversary then queries $m\oplus r_1$ on
  $\tilde{h}_R(\cdot)$ and can trivially distinguish $\tilde{h}$ from a
  random function.
%
\item $\tilde{h}_R(x)$ is defined as
  $\tilde{h}_1(x\oplus r_1) \oplus \tilde{h}_2 (x\oplus r_2)$. Now a
  slightly more complex attack can still succeed: the adversary
  subverts $h_1$ so that $\tilde{h}_1(x)=0$ if $x=m||*$, that is, when
  the first half of $x$ equals to a randomly selected string $m$ with
  length $n/2$; likewise, $h_2$ is subverted so that
  $\tilde{h}_2(x)=0$ if $x=*||m$, that is, the second half of $x$
  equals $m$. Then, the adversary queries $m_1||m_2$ on
  $\tilde{h}_R(\cdot)$, where $m_1=m\oplus r_{1,0}$, and
  $m_2=m\oplus r_{2,1}$, and $r_{1,0}$ is the first half of $r_1$, and
  $r_{2,1}$ is the second half of $r_2$.  Again, trivially, it can be
  distinguished from a random function.

  This attack can be generalized in a straightforward fashion to any
  $\ell\leq n/\lambda$: the input can be divided in into consecutive
  substrings each with length $\lambda$, and the ``trigger'' 
  substrings can be planted in each chunk.
%
\end{enumerate}




\medskip
\noindent\textbf{Challenges in the analysis.} To analyze security in
the ``crooked'' indifferentiability framework, our simulator needs to
ensure consistency between two ways of generating output values: one
is directly from the construction $C^{\tilde{H}^h}(x,R)$; the other
calls for an ``explanation'' of $F$---a truly random function---via
reconstruction from related queries to $H$ (in a way consistent with
the subverted implementation $\tilde{H}$). To ensure a correct
simulation, the simulator must suitably answer related queries
(defining one value of $C^{\tilde{H}^h}(x,R)$). Essentially, the proof
relies on an unpredictability property of the internal function
$\tilde{g}_R(x)$ to guarantee the success of simulation. In
particular, for any input $x$ (if not yet ``fully decided'' by
previous queries), the output of $\tilde{g}_R(x)$ is unpredicatable to
the distinguisher even if she knows the public randomness $R$ (even
conditioned on adaptive queries generated by $\sD$).

Section~\ref{sec:analysis} develops the detailed security analysis.
The proof of correctness for this construction is complicated by the
fact that the ``defining'' algorithm $\tilde{H}$ is permitted to make
adaptive queries to $h$ during the definition of $\tilde{h}$; in
particular, this means that even when a particular ``constellation''
of points (the $h_i(x \oplus r_i)$) contains a point that is left
alone by $\tilde{H}$---which is to say that it agrees with
$h_i()$---there is no guarantee that $\bigoplus_i h_i(x \oplus r_i)$
is uniformly random. This suggests focusing the analysis on
demonstrating that the constellation associated with every
$x \in \{0,1\}^n$ will have at least one ``good'' component, which is
(i.) not queried by $\tilde{H}^h(\cdot)$ when evaluated on the other
terms, and (ii.)  answered honestly. Unfortunately, actually
identifying such a good point with certainty appears to require that
we examine \emph{all} of the points in the constellation for $x$, and
this interferes with the standard ``exposure martingale'' proof that
is so powerful in the random oracle setting (which capitalizes on the
fact that ``unexamined'' values of $h$ can be treated as independent
and uniform values). Thus, even such elementary aspects of the
internal function require some care and must be handled with
appropriate union bounds. This points to a second difficulty, which is
that structural properties must generally be established with
exponentially small (that is $2^{-cn}$ for $c > 1$) error probability,
to permit a union bound over all inputs, or must be shown to fail with
negligible probability conditioned on a typical transcript of the
distinguisher; thus one must either establish very small error or
carry significant conditioning through the proof. In general, we aim
for the first of these options in order to simplify the
presentation. For example, part of analysis is supported with a
Fourier-analytic argument that establishes the near uniformity of the
random variable $x_1 \oplus \cdots \oplus x_\ell$ where each
$x_i \in \mathbb{Z}_2^n$ is drawn from a sufficiently rich multiset. This
provides a simple criterion for the initial random function $h$ from
which follows a number of strong properties that can conveniently
avoid the bookkeeping associated with transcript conditioning.
%

\smallskip
\noindent{\bf Applications:}
Our correction function can be easily applied to save the faulty hash implementation in several important application scenarios.
\begin{itemize}
\item Immediate applications, as explained in the motivational examples, one may apply our technique directly in the following scenarios to defend against backdoors in hash implementation.  

(1) {\em Salvaging PoW blockchain against hash subversion}. For proof-of-work based blockchains, 
as discussed above, miners may 
rely on a common library $\tilde{h}$ for the hash evaluation, perhaps
cleverly implemented by an adversary. Here $\tilde{h}$ is determined
before the chain has been deployed.  We can then prevent the adversary
from capitalizing on this subversion by applying our correction
function.
In particular, the public randomness $R$ can be embedded in the genesis
block; the function
$\tilde{h}_R(\cdot)$ is then used for mining (and verification) rather than
 $\tilde{h}$. 

\smallskip
(2) {\em Preventing system sneak-in.}
The system sneak-in can also be resolved
immediately by applying our correcting random oracle.
During system initialization (or even when the operating system is
released), the system administrator generates some randomness $R$ and
wraps the hash module $\tilde{h}$ (potentially subverted) to define
$\tilde{h}_R(\cdot)$. The password $\alpha$ then gives rise to the
digest $\rho=\tilde{h}_R(\alpha)$ together with the randomness
$R$. Upon receiving input $x$, the system first ``recovers''
$\tilde{h}_R(\cdot)$ based on the previously stored $R$, and then tests if
$\rho=\tilde{h}_R(x)$.  The access will be enabled if the test is
valid. As the corrected random oracle ensures the output to be uniform
for every input point,
this remains secure in the face of subversion.\footnote{Typical authentication of this form also uses password ``salt,'' but this doesn't change the structure of the attack or the solution.}

\item Extended applications:

\smallskip
(3) {\em Application to cliptographically secure digital signatures.} In a follow-up work, \cite{PKC:CRTYZZ19}, we constructed the first digital signature scheme that is secure against kleptographic attacks with only an offline watchdog  via {\em corrected} full domain hash. There, a further interpretation of our result is used that the corrected random oracle is indifferentiable also to a {\em keyed} random oracle.

\smallskip
(4) \noindent{\em Relation to random oracle against pre-processing attacks.} It is easy to see that our model is strictly stronger than the pre-processing model considered in \cite{EC:DodGuoKat17,EC:CDGS18}. There, the adversary made some  random oracle queries first and compress the transcripts as an auxiliary input. While in our model, besides this auxiliary input (which can be considered as part of the backdoor), the adversary is allowed to further subvert the implementation of the hash! It follows that our construction can also be directly used to defend the pre-processing attack.

\end{itemize}

\subsection{Related Work}

\noindent{\em Related work on indifferentiability.} The notion of
indifferentiability was proposed in the elegant work of Maurer et
al.~\cite{TCC:MauRenHol04}; this notably extends the classical concept
of indistinguishability to circumstances where one or more of the
relevant oracles are publicly available (such as a random oracle). It
was later adapted by Coron et al.~\cite{C:CDMP05}; several other
variants were proposed and studied in
\cite{TCC:DodPun06,EC:DodPun07}. A line of notable work applied the framework to to the 
ideal cipher problem: in particular the Feistel construction (with a
small constant number of rounds) is indifferentiable from an ideal
cipher, see \cite{JC:CHKPST16,EC:DacKatThi16,C:DaiSte16}. Our work
adapts the indifferentiability framework to the setting where the
construction uses only a subverted implementation
and the construction aims to be indifferentiable from a clean random
oracle.

\smallskip
\noindent{\em Related work on self-correcting programs. }
The theory of program self-testing, and self-correcting, was pioneered
by the work of Blum et
al.~\cite{Blum88,STOC:BluKan89,STOC:BluLubRub90}.  This theory
addresses the basic problem of program correctness by verifying
relationships between the outputs of the program on related, but
randomly selected, inputs; additionally, the theory studies
transformations of faulty programs that are almost correct (faulty at
negligible fraction of inputs) into ones that are correct at every
point with an overwhelming probability. See Rubinfeld's
thesis~\cite{Rubinfeld:1991} for great reading.  Our results can be
seen as a distributional version of this theory: (i.) we ``correct''
independent distributions rather than structured functions, (ii) we
insist on using only an exponentially small amount of public
randomness.

 \smallskip
\noindent{\em Related work on random oracles.}
The random oracle methodology/heuristic was first explicitly introduced by Bellare and Rogaway~\cite{CCS:BelRog93}; this methodology can significantly simplify both cryptographic constructions and proofs, even though  there exist schemes which are secure using random oracles, but cannot be instantiated in the standard model~\cite{STOC:CanGolHal98}.
Soon after, many separations have been shown for new classes of cryptographic tasks~\cite{C:Nielsen02a,C:Pass03,FOCS:GolKal03,TCC:CanGolHal04,EC:BelBolPal04,C:DodOliPie05,C:LeuNgu09,EC:KilPie09,TCC:BonSahWat11,SCN:GKMZ16}. 
In addition, efforts have  been made to identify instantiable assumptions/models in which we may analyze interesting cryptographic tasks~\cite{C:Canetti97,STOC:CanMicRei98,ICALP:CanDak08,AC:BCFW09,C:BolFis05,AC:BolFis06,C:KilONeSmi10,C:BelHoaKee13}.
Also, we note that research efforts have also been made to investigate
{\em weakened} idealized
models~\cite{PKC:NumIssTan08,PKC:KMTX10,SAC:Liskov06,RSA:KatLucThi15}. Finally,
there are several recent approaches that study random oracles in the
auxiliary input model (or with
pre-processing)~\cite{EC:DodGuoKat17,EC:CDGS18}. Our model is strictly
stronger than the pre-processing model: besides pre-processing
queries, for example, the adversary may embed some preprocessed
information into the subverted implementation; furthermore, our
subverted implementation can further misbehave in ways that cannot be
captured by any single-shot polynomial-query adversary because the
subversion at each point is determined by a local adaptive computation.
%

\smallskip
\noindent{\em Related work on kleptographic security.}  
Kleptographic attacks were originally introduced by Young and
Yung~\cite{C:YouYun96,EC:YouYun97}: In such attacks, the adversary
provides subverted implementations of the cryptographic primitive,
trying to learn secrets without being detected.
%
In recent years, several remarkable allegations of cryptographic tampering~\cite{PLS13,ReutersRSA2013}, including detailed investigations~\cite{Checkoway14,juniper_paper}, have produced a renewed interest in both kleptographic attacks and in techniques for preventing them~\cite{IEEESP:BBCL13,C:BelPatRog14,EC:DGGJR15,EC:MirSte15,C:DodMirSte16,FSE:DegFarPoe15,CCS:AteMagVen15,CCS:BelJaeKan15,EPRINT:SFKR15,EC:BelHoa15,CACM:AABB+15,C:DPSW16,AC:RTYZ16,CCS:RTYZ17,C:CamDriLeh17}. None of those work considered how to actually correct a subverted random oracle.

\smallskip
\noindent{\em Similar constructions in other context}. Our
construction follows a simple design approach, applying the hash
function to the XOR of multiple hash values. The construction calls to
mind many classical constructions for other purposes, e.g., hardness
amplification such as the Yao XOR lemma, weak PRF \cite{EC:Myers01},
and randomizers in the bounded storage model \cite{JC:DziMau04}. Our
construction has to have an ``extra'' layer of hash application ($h_0$
in our parlance) to wrap the XOR of terms, and our analysis is of
course very different from these classical results due to our 
starting point of a subverted implementation. 

\medskip
\noindent{\bf Relationship to the preliminary version.}
This paper expands on the preliminary conference version of this
article which appeared in CRYPTO '18~\cite{C:RTYZ18}. The conference
version provided a proof sketch which indicated how the major
technical hurdles can be overcome without developing the full
details. This article, while filling in those details, additionally
establishes somewhat tighter and simpler results. Specifically, in
many cases, we are able to improve the analysis to show that certain
events of interest fail with exponential, rather than negligible,
probability; this achieves tighter bounds and also simplifies the
presentation.

\medskip
\noindent{\bf Other follow-up work.} In~\cite{BNR20} Bhattacharyya et
al.\ pointed out a gap in the security sketch in our preliminary
version \cite{C:RTYZ18} and provide an alternate (and independent)
approach to prove that the construction of~\cite{C:RTYZ18} is
secure. They also explore an alternate, sponge-based construction that
reduces the number of external random bits to linear.

\ignore{
Regarding digital signatures, Ateniese et al.~\cite{CCS:AteMagVen15} first extended this line of study to digital signature scheme.  
Then, Russell et al.~\cite{AC:RTYZ16}  proposed defending mechanisms for digital signatures even against {\em subverted randomness generation}.
 However, these previous results \cite{AC:RTYZ16,CCS:AteMagVen15} essentially require an online watchdog which has to take all the transcript that communicated between the adversary and the challenger.  We remove this requirement and construct the first subversion resistant digital signature with only an offline watchdog. 
 }





\section{The Model: Crooked Indifferentiability}
\label{sec:model}


\subsection{Preliminary: Indifferentiability}
\label{sec:indifferentiability}

The notion of indifferentiability proposed in the elegant work of
Maurer et al.~\cite{TCC:MauRenHol04} has proven to be a powerful tool
for studying the security of hash function and many other primitives.
The notion extends the classical concept of indistinguishability to
the setting where one or more oracles involved in the construction are
publicly available. The indifferentiability framework
of~\cite{TCC:MauRenHol04} is built around random systems providing
interfaces to other systems. 
Coron et al.~\cite{C:CDMP05} demonstrate a strengthened\footnote{Technically, the quantifiers in the security definitions in the original~\cite{TCC:MauRenHol04} and in the followup ~\cite{C:CDMP05} are different; in the former, a simulator needs to be constructed for each adversary, while in the latter a simulator needs to be constructed for {\em all} adversaries.}
 indifferentiability framework built around Interactive
Turing Machines (as in~\cite{FOCS:Canetti01}). Our presentation
borrows heavily from~\cite{C:CDMP05}.  In the next subsection, we will
introduce our new notion, \emph{crooked indifferentiability}.

\paragraph{Defining indifferentiability.}
An \emph{ideal primitive} is an algorithmic entity which receives
inputs from one of the parties and returns its output immediately to
the querying party.  We now proceed to the definition of
indifferentiability~\cite{TCC:MauRenHol04,C:CDMP05}:


\begin{definition}[Indifferentiability~\cite{TCC:MauRenHol04,C:CDMP05}]
A Turing machine $\rC$ with oracle access to an ideal primitive $\cG$ is said to be   $(t_\cD,t_\cS,q,\epsilon)$-indifferentiable from an ideal primitive $\cF$, if there is a simulator $\cS$, such that for any distinguisher $\cD$, it holds that :
$$\left| \Pr[\cD^{\rC,\cG}(1^\secp) = 1] - \Pr[\cD^{\cF,\cS}(1^\secp) = 1] \right| \leq \epsilon\,.$$
  The simulator $\cS$ has oracle access to $\cF$ and runs in time at most $t_\cS$. The distinguisher $\cD$ runs in time at most $t_\cD$ and makes at most $q$ queries. Similarly, $\rC^\cG$ is said to be (computationally) indifferentiable from $\cF$ if $\epsilon$ is a negligible function of the security parameter $\secp$ (for polynomially bounded $t_\cD$ and $t_\cS$).
  See Figure \ref{fig:indiff}.
  
\end{definition}

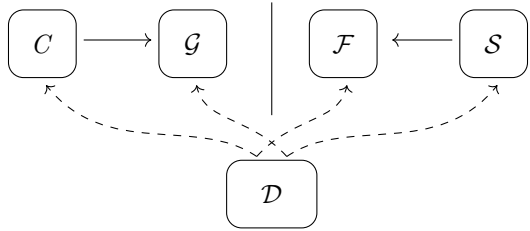
\begin{figure}[htbp!]
  \begin{center}
    \begin{tikzpicture}
            \draw[thin, rounded corners=2mm] (-3.5,1.5) rectangle +(.9,.9) node[pos=.5] {$\rC$};
      \draw[->,thin] (-2.5,2) -- (-1.6,2);

      \draw[thin, rounded corners=2mm] (-1.5,1.5) rectangle +(.9,.9) node[pos=.5] {$\cG$};
      \draw[thin] (0,2.5) -- (0,1);

                \draw[thin, rounded corners=2mm] (.5,1.5) rectangle +(.9,.9) node[pos=.5] {$\cF$};
      \draw[->,thin] (2.4,2) -- (1.6,2);

      \draw[thin, rounded corners=2mm] (2.5,1.5) rectangle +(.9,.9) node[pos=.5] {$\cS$};
           
\draw[thin,dashed,->] (0.2,0.45) .. controls (1,1) and (2.2,.4) .. (3,1.4);
      \draw[thin,dashed,->] (-0.2,0.45) .. controls (.3,1) and (.9,.9) .. (1,1.4);
\draw[thin,dashed,->] (0.2,0.45) .. controls (-.3,1) and (-.9,.9) .. (-1,1.4);
      \draw[thin,dashed,->] (-0.2,0.45) .. controls (-1,1) and (-2.2,.4) .. (-3,1.4);

            \draw[thin, rounded corners=2mm] (-.6,-.5) rectangle +(1.2,.9) node[pos=.5] {$\cD$};

    \end{tikzpicture}
  \end{center}
\caption{
\label{fig:indiff} 
The indifferentiability notion: the distinguisher $\cD$ either interacts with algorithm $\rC$ and ideal primitive $\cG$, or with ideal primitive $\cF$ and simulator $\cS$. Algorithm $\rC$ has oracle access to $\cG$, while simulator $\cS$ has oracle access to $\cF$.}
\end{figure}

As illustrated in Figure \ref{fig:indiff}, the role of the simulator is to simulate the ideal primitive
$\cG$ so that no distinguisher can tell whether it is interacting with $\rC$ and $\cG$, or with $\cF$ and
$\cS$; in other words, the output of $\cS$ should look ``consistent'' with what the distinguisher
can obtain from $\cF$. Note that the simulator does not observe the distinguisher's queries to
$\cF$; however, it can call $\cF$ directly when needed for the simulation. Note that, in some sense, the simulator must ``reverse engineer'' the construction $\rC$ so that the simulated oracle appropriately induces $\cF$ and, of course, possesses the correct marginal distribution. 

\paragraph{Replacement.}
 It is shown in \cite{TCC:MauRenHol04} that if $\rC^\cG$ is indifferentiable from $\cF$, then $\rC^\cG$ can replace $\cF$ in any cryptosystem, and the resulting cryptosystem is at least as secure in the  $\cG$  model as in the $\cF$ model. 

We use the definition of \cite{TCC:MauRenHol04} to specify what it means for a cryptosystem to be at least as secure in the  $\cG$  model as in the $\cF$ model. A cryptosystem is modeled as an Interactive Turing Machine with an interface to an adversary $\cA$ and to a public oracle. 
The cryptosystem is run by an environment $\cE$ which provides a binary output and also runs the adversary. In the  $\cG$  model, cryptosystem $\cP$ has oracle access to $\rC$ whereas attacker $\cA$ has oracle access to $\cG$. In the $\cF$ model, both $\cP$ and $\cA$ have oracle access to $\cF$. The definition is illustrated in Figure~\ref{fig:composition}.

\begin{figure}[htbp!]
  \begin{center}
    \begin{tikzpicture}
     
            \draw[thin, rounded corners=2mm] (-3.5,1.5) rectangle +(.9,.9) node[pos=.5] {$\rC$};

      \draw[->,thin] (-2.5,2) -- (-1.6,2);

      \draw[thin, rounded corners=2mm] (-1.5,1.5) rectangle +(.9,.9) node[pos=.5] {$\cG$};
      \draw[thin] (0,2.5) -- (0,-2.5);

                \draw[thin, rounded corners=2mm] (-3.5,0) rectangle +(.9,.9) node[pos=.5] {$\cP$};

      \draw[->,thin] (-2.5,0.5) -- (-1.6,0.5);
      
            \draw[->,thin] (-3,1) -- (-3,1.4);
            \draw[->,thin] (-1,1) -- (-1,1.4);

      \draw[thin, rounded corners=2mm] (-1.5,0) rectangle +(.9,.9) node[pos=.5] {$\cAtwo$};
             \draw[thin, rounded corners=2mm] (-3.5,-1.5) rectangle +(3,.9) node[pos=.5] {${\cE}$};
                                                                              \draw[->,thin] (-2,-1.5) -- (-2,-2.4);


      \draw[->,thin] (-2.5,0.5) -- (-1.6,0.5);
      
            \draw[->,thin] (-3,-.5) -- (-3,-0.1);
            \draw[->,thin] (-1,-.5) -- (-1,-.1);

                \draw[thin, rounded corners=2mm] (.5,1.5) rectangle +(.9,.9) node[pos=.5] {$\cF$};


                                  
                                     \draw[thin, rounded corners=2mm] (.5,0) rectangle +(.9,.9) node[pos=.5] {$\cP$};

      \draw[<->,thin] (-2.5,0.5) -- (-1.6,0.5);
      
            \draw[->,thin] (1,1) -- (1,1.4);
            \draw[->,thin] (3,1) |- (1.6,2);

      \draw[thin, rounded corners=2mm] (2.5,0) rectangle +(.9,.9) node[pos=.5] {$\cStwo$};

                      \draw[thin, rounded corners=2mm] (0.5,-1.5) rectangle +(3,.9) node[pos=.5] {${\cE}$};
                                                        \draw[->,thin] (2,-1.5) -- (2,-2.4);

      \draw[<->,thin] (1.5,0.5) -- (2.4,0.5);
      
           \draw[->,thin] (1,-.5) -- (1,-0.1);
            \draw[->,thin] (3,-.5) -- (3,-.1);

    \end{tikzpicture}
  \end{center}
\caption{
\label{fig:composition} 
The environment $\cE$ interacts with cryptosystem $\cP$ and attacker $\cAtwo$. In the $\cG$ model (left), $\cP$ has oracle access to $\rC$ whereas $\cAtwo$ has oracle access to $\cG$. In the $\cF$ model, both $\cP$ and $\cStwo$ have oracle access to $\cF$.}
\end{figure}
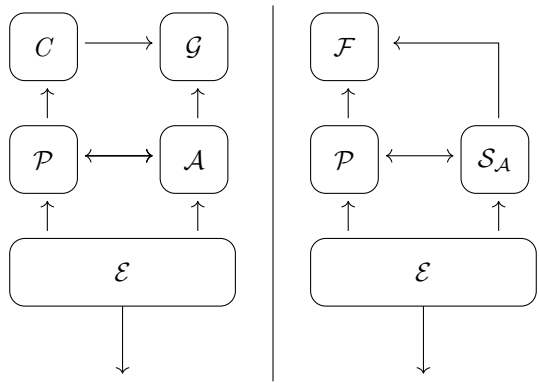

\begin{definition}
A cryptosystem is said to be at least as secure in the $\cG$ model with algorithm $\rC$ as in the $\cF$ model, 
if for any environment $\cE$ and any attacker $\cAtwo$ in the $\cG$ model, there exists an attacker $\cStwo$ in the $\cF$ model, such that: 
\[
\Pr[\cE(\cP^{\rC^{}},\cAtwo^{\cG})=1]-\Pr[\cE(\cP^\cF,\cStwo^\cF)=1]\leq\epsilon.
\]
where $\epsilon$ is a negligible function of the security parameter $\secp$. 
Similarly, a cryptosystem is said to be computationally at least as secure, etc., if $\cE$, $\cAtwo$ and $\cStwo$ are polynomial-time in $\secp$.
\end{definition}

 

We have the following security preserving (replacement) theorem, which says that when an ideal primitive is replaced by an indifferentiable one, the security of the ``big'' cryptosystem remains:
 \begin{theorem}[\cite{TCC:MauRenHol04,C:CDMP05}]
 \label{theorem:composition}
Let $\cP$ be a cryptosystem with oracle access to an ideal primitive $\cF$. Let $\rC$ be an algorithm such that $\rC^{\cG}$ is indifferentiable from $\cF$. Then cryptosystem $\cP$ is at least as secure in the $\cG$ model with algorithm $\rC$ as in the $\cF$ model. 
\end{theorem}

\subsection{Crooked indifferentiability}
\label{sec:crooked-indiff}

The ideal primitives that we focus on in this paper are random oracles.  
A random oracle~\cite{CCS:BelRog93} is an ideal primitive which provides an independent random output for each new query.  We next formalize a new notion called \emph{crooked indifferentiability} to reflect the challenges in our setting with subversion; our formalization is for random oracles, but the formalization can be naturally extended to other ideal primitives.    

\paragraph{Crooked indifferentiability for random oracles.}

As mentioned in the introduction, we consider the problem of ``repairing'' a subverted random oracle in such a way that the corrected construction can be used as a drop-in replacement for an unsubverted random oracle. This immediately suggests invoking and appropriately adapting the indifferentiability notion. Specifically, we need to adjust the notion to reflect subversion.

We model the act of \emph{subversion of a (hash) function $H$} as creation of an ``implementation'' $\tilde{H}$ of the new, subverted (hash) function; in practice, this would be the source code of the subverted version of the function $H$. In our setting, however, where $H$ is modeled as a random oracle, we define $\tilde{H}$ as a polynomial-time algorithm with oracle access to $H$; thus the subverted function is $x \mapsto \tilde{H}^H(x)$. 
We proceed to survey the main modifications between crooked indifferentiability and the original notion of indifferentiability. 
\begin{enumerate}
\item The deterministic construction will have oracle access to the random oracle only via the subverted implementation $\tilde{H}$ but not via the ideal primitive $H$. (Operationally, the construction has oracle access to the function $x \mapsto \tilde{H}^H(x)$.)
The construction depends on access to trusted, but public, randomness $R$.  
\item The simulator is provided, as input, the subverted implementation $\tilde{H}$ (a Turing machine) and the public randomness $R$; it has oracle access to the target ideal functionaltiy ($\cF$).
\end{enumerate}
Point (2) is necessary, and desirable, as it is clearly impossible to achieve indifferentiability using a simulator that has no access to $\tilde{H}$ (the distinguisher can simply query an input such that $\rC$ will use a value that is modified by $\tilde{H}$ while $\cS$ has no way to reproduce this). More importantly, we will show below that security will be preserved by replacing an ideal random oracle with a construction satisfying our definition (with an augmented simulator). Specifically, we prove a security preserving (i.e., replacement) theorem akin to those of \cite{TCC:MauRenHol04} and \cite{C:CDMP05} for our adapted notions. 

\begin{definition}[$H$-crooked indifferentiability] 
\label{def:indiff-crooked}
Consider a distinguisher $\sD$ and the following multi-phase real
execution. Initially, the distinguisher $\sD$ commences the first
phase: with oracle access to ideal primitive $H$ the distinguisher
constructs and publishes a \emph{subverted implementation} of $H$;
this subversion is described as a deterministic polynomial time
algorithm denoted $\tilde{H}$. (Recall that the algorithm $\tilde{H}$
implicitly defines a subverted version of $H$ by providing $H$ to
$\tilde{H}$ as an oracle---thus $\tilde{H}^H(x)$ is the value taken by
the subverted version of $H$ at $x$.)  Then, a uniformly random string
$R$ is sampled and published.  Then the second phase begins involving
a deterministic construction $\rC$: the construction
$\rC$ requires the random string $R$ as input and has oracle access to
$\tilde{H}$ (the crooked version of $H$); explicitly this is the
oracle $x \mapsto \tilde{H}^H(x)$.  Finally, the distinguisher $\sD$,
now with random string $R$ as input and full oracle access to the pair
$(\rC, H)$, returns a decision bit $b$.  Often, we call $\sD$ the
$H$-crooked-distinguisher.

Consider now the corresponding multi-phase ideal execution with the same $H$-crooked-distinguisher $\sD$. The ideal execution introduces a simulator $\cS$ responsible for simulating the behavior of $H$; the simulator is provided full oracle access to the ideal object $\cF$. Initially, the simulator must answer any queries made to $H$ by $\sD$ in the first phase. Then the simulator is given the random string $R$ and the algorithm $\langle \tilde{H}\rangle$ (generated by $\sD$ at the end of the first phase) as input. In the second phase, the $H$-crooked-distinguisher $\sD$, now with random string $R$ as input and oracle access to the alternative pair $(\cF, \cS)$, returns a decision bit $b$. 

We say that construction $\rC$ is
$(n_{\source},n_{\target},q_{\sD},q_{\tilde{H}},r,\epsilon)$-$H$-crooked-indifferentiable
from ideal primitive $\cF$ if there is an efficient simulator $\cS$ so
that for any $H$-crooked-distinguisher $\sD$ making no more than
$q_{\sD}(\secp)$ queries and producing a subversion $\tilde{H}$ making
no more than $q_{\tilde{H}}(\secp)$ queries, the real execution and
the ideal execution are indistinguishable.  Specifically,
\[
  \left| \Pr_{u,R,H} \left[\tilde{H} \leftarrow\sD^H(1^\secp)\ ; \
      \sD^{\rC^{\tilde{H}}(R),H}(1^\secp, R) = 1\right] -
    \Pr_{u,R,\cF} \left[\tilde{H} \leftarrow\sD^H(1^\secp)\ ; \
      \sD^{\cF,\cS_{}^{\cF}(R,\langle \tilde{H}\rangle)}(1^\secp, R) = 1\right]
  \right| \leq \epsilon(\secp)\,.
\]
Here $R$ denotes a random string of length $r(\secp)$ and both
$H: \bool^{n_{\source}} \rightarrow \bool^{n_{\source}}$ and
$\cF: \bool^{n_{\target}} \rightarrow \bool^{n_{\target}}$ denote
random functions where $n_{\source}(\secp)$ and $n_{\target}(\secp)$
are polynomials in the security parameter $\secp$. We let $u$ denote
the random coins of $\sD$. The simulator is efficient in the sense
that it is polynomial in $n$ and the running time of the supplied
algorithm $\tilde{H}$ (on inputs of length $n_{\source}$). See
Figure~\ref{fig:indiff-crooked} for detailed illustration of the last
phase in both real and ideal executions. (While it is not explicitly
captured in the description above, the distinguisher $\sD$ is
permitted to carry state from the first phase to the second phase.)
The notation $C^{\tilde{H}}(R)$ denotes oracle access to the function
$x \mapsto \tilde{H}(x)$.
\end{definition}

Our main security proof will begin by demonstrating that in our
particular setting, security in a simpler model suffices: this is the
\emph{abbreviated crooked indifferentiability} model, articulated
below. We then show that---in light of the special structure of our
simulator---it can be effectively lifted to the full model above.

\begin{definition}[Abbreviated $H$-crooked indifferentiability]
  \label{def:abbrev-indiff-crooked}
  The abbreviated model calls for the distinguisher to provide the
  subversion algorithm $\tilde{H}$ at the outset (without the
  advantage of any preliminary queries to $H$). Thus, the abbreviated
  model consists only of the last phase of the full model.  Formally,
  in the abbreviated model the distinguisher is provided as a pair
  $(\sD,\tilde{H})$, the random string $R$ is drawn (as in the full
  model), and insecurity is expressed as the difference between the
  behavior of $\sD$ on the pair $(C^{\tilde{H}}(R),H)$ and the pair
  $(\cF,\cS^\cF(R,\langle\tilde{H}\rangle))$. Specifically, the
  construction $\rC$ is
  $(n_{\source},n_{\target},q_{\sD},q_{\tilde{H}},r,\epsilon)$-Abbreviated-$H$-crooked-indifferentiable
  from ideal primitive $\cF$ if there is an efficient simulator $\cS$
  so that for any $H$-crooked-distinguisher $\sD$ making no more than
  $q_{\sD}(\secp)$ queries and subversion algorithm $\tilde{H}$ making
  no more than $q_{\tilde{H}}(\secp)$ queries, the real execution and
  the ideal execution are indistinguishable:
  \[
    \left| \Pr_{u,R,H} \left[\sD^{\rC^{\tilde{H}}(R),H}(1^\secp, R) = 1\right] - 
      \Pr_{u,R,\cF} \left[\sD^{\cF,\cS_{}^{\cF}(R,\langle\tilde{H}\rangle)}(1^\secp, R) = 1\right] 
    \right| \leq \epsilon(\secp)\,.
  \]
  Here $R$ denotes a random string of length $r(\secp)$ and both
  $H: \bool^{n_{\source}} \rightarrow \bool^{n_{\source}}$ and
  $\cF: \bool^{n_{\target}} \rightarrow \bool^{n_{\target}}$ denote
  random functions where $n_{\source}(\secp)$ and $n_{\target}(\secp)$
  are polynomials in the security parameter $\secp$. We let $u$ denote
  the random coins of $\sD$. The simulator is efficient in the sense
  that it is polynomial in $n$ and the running time of the supplied
  algorithm $\tilde{H}$ (on inputs of length $n_{\source}$).
\end{definition}

Observe that while the abbreviated simulator is a fixed algorithm, its
running time may depend on the running time of $\tilde{H}$---in
particular, the definition permits $\cS$ sufficient running time to
simulate $\tilde{H}$ on a polynomial number of inputs.

Regarding the difference between these notions, observe that the
distinguisher can ``compile into'' the subversion algorithm
$\tilde{H}$ any queries and pre-computation that might have been
advantageous to carry out in phase I; such queries and pre-computation
can also be mimicked by the distinguisher itself. This technique can
effectively simulate the two phase execution with a single
phase. Nevertheless, the models do make slightly different demands on
the simulator which must be prepared to answer some queries (in Phase
I) prior to knowledge of $R$ and $\tilde{H}$.

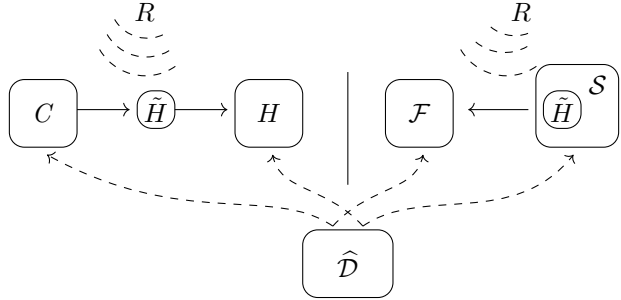
\begin{figure}[htbp!]
  \begin{center}
    \begin{tikzpicture}
            \draw[thin, rounded corners=2mm] (-4.5,1.5) rectangle +(0.9,0.9) node[pos=.5] {$\rC$} node[pos=2] {$R$};
                        \draw[thin, rounded corners=2mm] (-2.8,1.75) rectangle +(.5,.5) node[pos=.5] {$\tilde{H}$};

      \draw[->,thin] (-3.6,2) -- (-2.9,2);
      \draw[->,thin] (-2.3,2) -- (-1.6,2);

      \draw[thin, rounded corners=2mm] (-1.5,1.5) rectangle +(.9,.9) node[pos=.5] {$H$};
      \draw[thin] (0,2.5) -- (0,1);

                \draw[thin, rounded corners=2mm] (.5,1.5) rectangle +(.9,.9) node[pos=.5] {$\cF$} node[pos=2] {$R$};
      \draw[->,thin] (2.4,2) -- (1.6,2);

\draw [thin,dashed] (-3.1,3.2) arc [radius=.5, start angle=210, end angle= 300]; 

\draw [thin,dashed] (-3.2,3) arc [radius=.6, start angle=210, end angle= 310]; 

\draw [thin,dashed] (-3.3,2.8) arc [radius=.7, start angle=210, end angle= 310]; 

\draw [thin,dashed] (1.9,3.1) arc [radius=.5, start angle=230, end angle= 300]; 

\draw [thin,dashed] (1.7,2.9) arc [radius=.6, start angle=230, end angle= 300]; 

\draw [thin,dashed] (1.5,2.8) arc [radius=.7, start angle=210, end angle= 310]; 

      \draw[thin, rounded corners=2mm] (2.5,1.5) rectangle +(1.1,1.1) node[pos=.75] {$\cS$};
                                  \draw[thin, rounded corners=2mm] (2.6,1.75) rectangle +(.5,.5) node[pos=.5] {$\tilde{H}$};
 
\draw[thin,dashed,->] (0.2,0.45) .. controls (1,1) and (2.2,.4) .. (3,1.4);5
      \draw[thin,dashed,->] (-0.2,0.45) .. controls (.3,1) and (.9,.9) .. (1,1.4);
\draw[thin,dashed,->] (0.2,0.45) .. controls (-.3,1) and (-.9,.9) .. (-1,1.4);
      \draw[thin,dashed,->] (-0.2,0.45) .. controls (-1,1) and (-2.2,.4) .. (-4,1.4);

            \draw[thin, rounded corners=2mm] (-.6,-.5) rectangle +(1.2,.9) node[pos=.5] {$\sD$};

    \end{tikzpicture}
  \end{center}
  \caption{
\label{fig:indiff-crooked} 
The $H$-crooked indifferentiability notion: the distinguisher $\sD$, in the first phase, manufactures and publishes a subverted implementation denoted as $\tilde{H}$, for ideal primitive $H$; then in the second phase, a random string $R$ is published; 
after that, in the third phase, algorithm $\rC$, and simulator $\cS$ are developed; 
the $H$-crooked-distinguisher $\sD$, in the last phase,   
either interacting with algorithm $\rC$ and ideal primitive $H$, or with ideal primitive $\cF$ and simulator $\cS$, return a decision bit. 
Here, algorithm $\rC$ has oracle access to $\tilde{H}$, while simulator $\cS$ has oracle access to $\cF$ and $\tilde{H}$.
}
\end{figure}

\paragraph{Replacement with crooked indifferentiability.} 

The notion of crooked indifferentiability 
is formalized for a specific ideal primitive, i.e., random oracles. We note that the formalization can be trivially generalized for all ideal primitives.  

Security preserving (replacement) has been shown in the indifferentiability framework~\cite{TCC:MauRenHol04}: if $\rC^\cG$ is indifferentiable from $\cF$, then $\rC^\cG$ can replace $\cF$ in any cryptosystem, and the resulting cryptosystem in the  $\cG$  model is at least as secure  as that in the $\cF$ model. 
We next show that the replacement property also holds in the crooked indifferentiability framework. 

 Recall that, in the ``standard'' indifferentiability framework~\cite{TCC:MauRenHol04,C:CDMP05},  a cryptosystem can be modeled as an Interactive Turing Machine with an interface to an adversary $\cA$ and to a public oracle. 
There the cryptosystem is run by a  ``standard'' environment $\cE$.  
In our ``crooked'' indifferentiability framework, a cryptosystem has the interface to an adversary $\cA$ and to a public oracle. However, now  the cryptosystem is run by a  crooked-environment $\sE$.

Consider an ideal primitive $\cG$. 
Similar to the $\cG$-crooked-distinguisher, we can define the $\cG$-crooked-environment $\sE$ as follows: 
Initially, the crooked-environment $\sE$ manufactures and then publishes a subverted implementation of the ideal primitive $\cG$, and denotes it $\badG$. Then $\sE$ runs the attacker $\cAtwo$, and 
the cryptosystem $\cP$ is developed.   
In the  $\cG$  model, cryptosystem $\cP$ has oracle access to $\rC$ whereas attacker $\cAtwo$ has oracle access to $\cG$; note that, $\rC$ has oracle access to $\badG$, not to directly $\cG$. In the $\cF$ model, both $\cP$ and $\cAtwo$ have oracle access to $\cF$. Finally, the crooked-environment $\sE$ returns a binary decision output. 
The definition is illustrated in Figure \ref{fig:composition-crooked}.


\begin{definition}
Consider ideal primitives $\cG$ and $\cF$. 
A cryptosystem $\cP$ is said to be at least as secure in the $\cG$-crooked model with algorithm $\rC$ as in the $\cF$ model, if for any $\cG$-crooked-environment $\sE$ and any attacker $\cAtwo$ in the $\cG$-crooked model, there exists an attacher $\cStwo$ in the $\cF$ model, such that: 
$$\Pr[\sE(\cP^{\rC^{\badG}},\cAtwo^{\cG})=1]-\Pr[\sE(\cP^\cF,\cStwo^\cF)=1]\leq\epsilon.$$
where $\epsilon$ is a negligible function of the security parameter $\secp$. 
\end{definition}

\begin{figure}[htbp!]
 \begin{center}
 \begin{tikzpicture}

 \draw[thin, rounded corners=2mm] (-3.5,1.5) rectangle +(.9,.9) node[pos=.5] {$\rC$} node[pos=2] {$R$};
 \draw[thin, rounded corners=2mm] (-2.25,1.75) rectangle +(.5,.5) node[pos=.5] {$\badG$};

 \draw[->,thin] (-1.8,2) -- (-1.5,2);
 \draw[->,thin] (-2.6,2) -- (-2.25,2);

 \draw[thin, rounded corners=2mm] (-1.5,1.5) rectangle +(.9,.9) node[pos=.5] {$\cG$};
 \draw[thin] (0,2.5) -- (0,-2.5);
 
 \draw[thin, rounded corners=2mm] (-3.5,0) rectangle +(.9,.9) node[pos=.5] {$\cP$};

 \draw[->,thin] (-2.5,0.5) -- (-1.6,0.5);
 
 \draw[->,thin] (-3,1) -- (-3,1.4);
 \draw[->,thin] (-1,1) -- (-1,1.4);

\draw [thin,dashed] (-2.1,3.2) arc [radius=.5, start angle=210, end angle= 300]; 

\draw [thin,dashed] (-2.2,3) arc [radius=.6, start angle=210, end angle= 310]; 

\draw [thin,dashed] (-2.3,2.8) arc [radius=.7, start angle=210, end angle= 310]; 

\draw [thin,dashed] (2.1,3.1) arc [radius=.5, start angle=230, end angle= 300]; 

\draw [thin,dashed] (1.9,2.9) arc [radius=.6, start angle=230, end angle= 300]; 

\draw [thin,dashed] (1.7,2.8) arc [radius=.7, start angle=210, end angle= 310]; 

 \draw[thin, rounded corners=2mm] (-1.5,0) rectangle +(.9,.9) node[pos=.5] {$\cAtwo$};

 \draw[thin, rounded corners=2mm] (-3.5,-1.5) rectangle +(3,.9) node[pos=.5] {${\sE}$};
 
 \draw[->,thin] (-2,-1.5) -- (-2,-2.4);


 \draw[->,thin] (-2.5,0.5) -- (-1.6,0.5);
 
 \draw[->,thin] (-3,-.5) -- (-3,-0.1);
 \draw[->,thin] (-1,-.5) -- (-1,-.1);
 


 

 \draw[thin, rounded corners=2mm] (.5,1.5) rectangle +(.9,.9) node[pos=.5] {$\cF$} node[pos=2] {$R$};


 
 \draw[thin, rounded corners=2mm] (.5,0) rectangle +(.9,.9) node[pos=.5] {$\cP$};

 \draw[<->,thin] (-2.5,0.5) -- (-1.6,0.5);
 
 \draw[->,thin] (1,1) -- (1,1.4);
 \draw[->,thin] (3,1) |- (1.6,2);

 \draw[thin, rounded corners=2mm] (2.5,0) rectangle +(.9,.9) node[pos=.5] {$\cS_\cA$};

 \draw[thin, rounded corners=2mm] (0.5,-1.5) rectangle +(3,.9) node[pos=.5] {${\sE}$};
 \draw[->,thin] (2,-1.5) -- (2,-2.4);

 \draw[<->,thin] (1.5,0.5) -- (2.4,0.5);
 
 \draw[->,thin] (1,-.5) -- (1,-0.1);
 \draw[->,thin] (3,-.5) -- (3,-.1);

 \end{tikzpicture}
 \end{center}
 \caption{
\label{fig:composition-crooked} 
The environment $\sE$ interacts with cryptosystem $\cP$ and attacker $\cAtwo$:
In the $\cG$ model (left), $\cP$ has oracle accesses to $\rC$ whereas $\cAtwo$ has oracle accesses to $\cG$; the algorithm $\rC$ has oracle accesses to the subverted $\badG$. 
In the $\cF$ model, both $\cP$ and $\cStwo$ have oracle accesses to $\cF$. 
In addition, in both $\cG$ and $\cF$ models, randomness $R$ is publicly available to all entities.}
\end{figure}



 We now demonstrate the following theorem which 
 shows that security is preserved when replacing an ideal primitive by a crooked-indifferentiable one:

 \begin{theorem} \label{theorem:crooked-composition}
 Consider an ideal primitive $\cG$ and a $\cG$-crooked-environment $\sE$. 
Let $\cP$ be a cryptosystem with oracle access to an ideal primitive $\cF$. Let $\rC$ be an algorithm such that $\rC^{\cG}$ is $\cG$-crooked-indifferentiable from $\cF$. Then cryptosystem $\cP$ is at least as secure in the $\cG$-crooked model with algorithm $\rC$ as in the $\cF$ model. 
\end{theorem}

\begin{proof}
The proof is very similar to that in~\cite{TCC:MauRenHol04,C:CDMP05}. 
Let $\cP$ be any cryptosystem, modeled as an Interactive Turing Machine. Let $\sE$ be any crooked-environment, and $\cAtwo$ be any attacker in the $\cG$-crooked model. In the $\cG$-crooked model, $\cP$ has oracle access to $\rC$ (who has oracle access to $\badG$, not to directly $\cG$.)
 whereas $\cAtwo$ has oracle access to ideal primitive $\cG$; moreover crooked-environment $\sE$ interacts with both $\cP$ and $\cAtwo$. This is illustrated in Figure \ref{fig:composition-crooked-proof} (left part).

Since $\rC$ is crooked-indifferentiable from $\cF$ (see Figure \ref{fig:indiff-crooked}), one can replace $(\rC^{\badG}, \cG)$ by $(\cF, \cS)$ with only a negligible modification of the crooked-environment $\sE$'s output distribution. As illustrated in Figure \ref{fig:composition-crooked-proof}, by merging attacker $\cAtwo$ and simulator $\cS$, one obtains an attacker $\cStwo$ in the $\cF$ model, and the difference in $\sE$'s output distribution is negligible.
\end{proof}

\begin{figure}[htbp!]
 \begin{center}
 \begin{tikzpicture}
 \draw[thin, rounded corners=2mm] (-3.5,1.5) rectangle +(.9,.9) node[pos=.5] {$\rC$} node[pos=2] {$R$};
 \draw[thin, rounded corners=2mm] (-2.25,1.75) rectangle +(.5,.5) node[pos=.5] {$\badG$} ;

 \draw[->,thin] (-1.8,2) -- (-1.5,2);
 \draw[->,thin] (-2.6,2) -- (-2.25,2);

 \draw[thin, rounded corners=2mm] (-1.5,1.5) rectangle +(.9,.9) node[pos=.5] {$\cG$};
 \draw[thin] (0,2.5) -- (0,-2.5) ;
 
 
\draw [thin,dashed] (-2.1,3.2) arc [radius=.5, start angle=210, end angle= 300]; 

\draw [thin,dashed] (-2.2,3) arc [radius=.6, start angle=210, end angle= 310]; 

\draw [thin,dashed] (-2.3,2.8) arc [radius=.7, start angle=210, end angle= 310]; 

\draw [thin,dashed] (2.1,3.1) arc [radius=.5, start angle=230, end angle= 300]; 

\draw [thin,dashed] (1.9,2.9) arc [radius=.6, start angle=230, end angle= 300]; 

\draw [thin,dashed] (1.7,2.8) arc [radius=.7, start angle=210, end angle= 310]; 

 \draw[thin, rounded corners=2mm] (-3.5,0) rectangle +(.9,.9) node[pos=.5] {$\cP$};

 \draw[->,thin] (-2.5,0.5) -- (-1.6,0.5);
 
 \draw[->,thin] (-3,1) -- (-3,1.4);
 \draw[->,thin] (-1,1) -- (-1,1.4);

 \draw[thin, rounded corners=2mm] (-1.5,0) rectangle +(.9,.9) node[pos=.5] {$\cAtwo$};

 \draw[thin, rounded corners=2mm] (-3.5,-1.5) rectangle +(3,.9) node[pos=.5] {${\sE}$};
 
 \draw[->,thin] (-2,-1.5) -- (-2,-2.4);


 \draw[->,thin] (-2.5,0.5) -- (-1.6,0.5);
 
 \draw[->,thin] (-3,-.5) -- (-3,-0.1);
 \draw[->,thin] (-1,-.5) -- (-1,-.1);
 
 \draw[thin, dashed, rounded corners=2mm] (-3.7,-2.1) rectangle +(3.4,3.3) node[pos=.1] {$\sD$};



 \draw[thin, rounded corners=2mm] (.5,1.5) rectangle +(.9,.9) node[pos=.5] {$\cF$} node[pos=2] {$R$};
 \draw[->,thin] (2.4,2) -- (1.6,2);

 \draw[thin, rounded corners=2mm] (2.5,1.5) rectangle +(1.1,1.1) node[pos=.75] {$\cS$};
 \draw[thin, rounded corners=2mm] (2.6,1.75) rectangle +(.5,.5) node[pos=.5] {$\badG$};

 \draw[thin, rounded corners=2mm] (.5,0) rectangle +(.9,.9) node[pos=.5] {$\cP$};

 \draw[<->,thin] (-2.5,0.5) -- (-1.6,0.5);
 
 \draw[->,thin] (1,1) -- (1,1.4);
 \draw[->,thin] (3,1) -- (3,1.4);

 \draw[thin, rounded corners=2mm] (2.5,0) rectangle +(.9,.9) node[pos=.5] {$\cAtwo$};

 \draw[thin, rounded corners=2mm] (0.5,-1.5) rectangle +(3,.9) node[pos=.5] {${\sE}$};
 \draw[->,thin] (2,-1.5) -- (2,-2.4);


 \draw[<->,thin] (1.5,0.5) -- (2.4,0.5);
 
 \draw[->,thin] (1,-.5) -- (1,-0.1);
 \draw[->,thin] (3,-.5) -- (3,-.1);
 
 \draw[thin, dashed, rounded corners=2mm] (.3,-2.1) rectangle +(3.4,3.3) node[pos=.1] {$\sD$};
 
 
 \draw[thin, dashed,red, rounded corners=2mm] (2.3,-.3) rectangle +(1.9,3) node[pos=.85] {$\cStwo$};



 \end{tikzpicture}
 \end{center}
\caption{
\label{fig:composition-crooked-proof} 
Construction of attacker $\cStwo$ from attacker $\cAtwo$ and simulator $\cS$}
\end{figure}



\paragraph{Restrictions (of using crooked indifferentiability).}
Ristenpart et al.~\cite{EC:RisShaShr11} has demonstrated that the replacement/composition theorem (Theorem~\ref{theorem:composition}) in the original indifferentiability framework only holds in single-stage settings. We remark that, the same restriction also applies to our replacement/composition theorem (Theorem~\ref{theorem:crooked-composition}). We leave it as our future work to extend our crooked indifferentiability to the multi-stage settings where disjoint adversaries are split over several stages.




\section{The Construction}
\label{sec:construction}
From this point on, we use $\cD$ rather than $\sD$ to denote the
distinguisher in our crooked indifferentiability model. For a security
parameter $n$ and a (polynomially related) parameter $\ell$, the
construction depends on public randomness $R = (r_1, \ldots, r_\ell)$,
where each $r_i$ is an independent and uniform element of
$\{0,1\}^n$.

The source function of the construction is expressed as a family of
$\ell+1$ independent random oracles:
\begin{align*}
  h_0&: \{0,1\}^{n} \rightarrow \{0,1\}^{3n}\,,\\
  h_i&: \{0,1\}^{3n} \rightarrow \{0,1\}^{n}\,,\qquad\text{for
       $i \in \{1, \ldots, \ell\}$.}
\end{align*}
These can be realized as slices of a single random function
$H: \{0,1\}^{n'} \rightarrow \{0,1\}^{n'}$, with
$n' = 3n + \lceil \log \ell+1 \rceil$ by an appropriate convention for
embedding and extracting inputs and values. In the few cases where we
need to be precise, we write
$\{0, \ldots, L\} \times \{0,1\}^{3n} = \{0,1\}^{n'}$, where
$\ell+1 \leq L = 2^{\lceil \log \ell+1 \rceil}$, and let $[i,x]$
denote a unique element of $\{0, \ldots, L\} \times \{0,1\}^{3n}$ to
implicitly correspond to the element $x$ in the domain for $h_i$: for
concreteness, define $[0,x] = (0,0^{2n}x)$ and $[i,x] = (i, x)$ for
$i>0$. The output of $H$ is treated as a string of the correct length
without any special indication. Given subverted implementations
$\{\tilde{h}_i\}_{i=0,\ldots,\ell}$ (defined as above by the
adversarially-defined algorithm $\tilde{H}$),
the corrected function is defined as:
\[
  \rC^{\tilde{H}^H}(x) \eqdef \tilde{h}_0\left(\bigoplus_{i = 1}^{\ell}
    \tilde{h}_i(x \oplus r_i)\right)\,,
\]
where $R = (r_1, \ldots, r_\ell)$ is sampled uniformly after
$\tilde{h}$ is provided (and then revealed to the public). We refer to
this function $\rC$ as $\tilde{h}_R(x)$ and, for the purposes of
analysis, also give a name to the ``internal function''
$\tilde{g}_R(\cdot)$:
\[
\tilde{g}_R(x) = \bigoplus_{i = 1}^{\ell} \tilde{h}_i(x \oplus r_i).
\]
Thus
\[
  \tilde{h}_R(x) \eqdef  \tilde{h}_0(\tilde{g}_R(x)) = \tilde{h}_0\left(\bigoplus_{i = 1}^{\ell} \tilde{h}_i(x \oplus r_i)\right)\,.
\]
We wish to show that such a construction is indifferentiable from an actual random oracle (with the proper input/output length). This implies, in particular, that
values taken by $\tilde{h}_R(\cdot)$ at
inputs that have not been queried have negligible distance from the
uniform distribution.

\begin{theorem}
\label{thm:ind}
We treat a function $h:\{0,1\}^{n'} \rightarrow \{0,1\}^{n'}$, with
$n' = 3n + \lceil \log \ell+1 \rceil$, as implicitly defining a family
of random oracles
\begin{align*}
  h_0& :\{0,1\}^{3n}\rightarrow\{0,1\}^{n}\,, \qquad\text{and}\\
  h_i& :\{0,1\}^n\rightarrow\{0,1\}^{3n}\,, \qquad\text{for $i>0$,}
\end{align*}
by treating $\{0,1\}^{n'} = \{0, \ldots, L-1\} \times \{0,1\}^{3n}$
and defining $h_i(x) = h(i,\cdot)$, for $i=0,\ldots,\ell \leq L-1$.
(Output lengths are achieved by removing the appropriate number of
trailing symbols). We will use the setting $\ell > n+4$.  Consider
a (subversion) algorithm $\tilde{H}$ so that
  $\tilde{H}^h(x)$ defines a subverted random oracle $\tilde{h}$. Assume that for every $h$ (and every $i$),  
  \begin{equation}\label{eq:equal}
    \Pr_{x \in \{0,1\}^n}[\tilde{h}(i,x) \neq h(i,x)] \leq \epsilon(n) = \negl(n)\,.
  \end{equation}

The construction $\tilde{h}_R(\cdot)$ is $(n',n,q_{\cD},q_{\tilde{H}},r,\epsilon')$-indifferentiable from a random oracle $F:\{0,1\}^{n}\rightarrow\{0,1\}^n$, where $\epsilon'=O(\ell q_{\hat{D}}q_{\tilde{H}}/\sqrt{2^n}+\sqrt{\ell}q_{\hat{D}}\epsilon^{1/16}).$, $q_{\cD}$ is 
the number of queries made by the distinguisher $\cD$ and $q_{\tilde{H}}$ is 
the number of queries made by $\tilde{H}$ as in Definition \ref{def:indiff-crooked}. $q_{\cD}$ and $q_{\tilde{H}}$ are both polynomial functions of $n$.

\end{theorem}

\paragraph{Domain extension.} It is shown in~\cite{C:CDMP05} that an arbitrary-length random oracle can be constructed from a fixed-length one in the indifferentiability framework, thus our result can be easily generalized to handle the case for correcting a subverted arbitrary input length hash.

\medskip
\noindent{\bf Roadmap for the proof.} We first describe the simulator algorithm. The main challenge for the simulator is to ensure the consistency of two ways of generating the output values of the construction.
The idea for simulation is fairly simple: the $h_i()$, for $i > 0$,
are treated as random functions; $h_0$ is programmed to suitably agree
with $F$. Specifically, once the value $\tilde{g}_R(x)$ is determined,
the value $F(x)$ is used to program $h_0$ on input
$\tilde{g}_R(x)$. Our simulator eagerly ``completes'' any
constellation of related points once any of its constituent points
have been queried by the distinguisher; in particular, any query to
$h_i(x \oplus r_i)$ prompts the simulator to evaluate
$\tilde{h}_j(x \oplus r_j)$ for all $j$; this provides enough
information to properly program $h_0(\cdot)$ at the desired location
$\tilde{g}_R(x)$ to ensure consistency with $F(\cdot)$. This convention
leads to a relatively simple invariant maintained during a typical
interaction with the distinguisher: all constellations thus far
touched by the distinguisher are correctly programmed. Of course,
proving that this simulation doesn't run into consistency
problems---for example, a circumstance where $h_0(y)$ has
unfortunately already been assigned a value at the moment it would be
appropriate to program it---is part of the proof of correctness.


There are two fundamental obstacles that hinder the simulation: (1)
for some $x$, after completing the constellation associated with $x$,
the simulator finds that $h_0$ has been previously queried on
$\tilde{g}_R(x)$---thus the simulator's hands are tied when it comes
time to program this value of $h_0$ (to be consistent with $F$);
(2) the distinguisher queries some input $x$ such that
$\tilde{g}_R(x)$ falls into the incorrect (subverted) portion of
$\tilde{h}_0$. It's convenient to further separate various specific
patterns of behavior that can give rise to collisions as described in
(1) above. These are reflected in a sequence of four ``hybrid games'' that interpolate between the two interactions of interest---the result of the distinguisher when interacting with the construction and the result when interacting with the simulator.

To further simplify the proof, we initially focus on the abbreviated
variant of crooked indifferentiability where the distinguisher does
not have the luxury of making preliminary queries to the hash function
while it formulates its subversion algorithm $\tilde{H}$. This places
all queries of the distinguisher on the same footing and somewhat
simplifies bookkeeping. We then return at the conclusion of the main
proof show that the same techniques can be applied to the full
setting.

\section{Security Proof}
\label{sec:analysis}
We begin with an abstract formulation of the properties of our construction, and then transition to the detailed description of the simulator algorithm and its effectiveness.

Anticipating the full description of the simulator, we formally set
down the notion of a ``constellation'' of points.
\begin{definition}[Constellation]
  For a fixed $R = (r_1, \ldots, r_\ell)$ and an element
  $x \in \{0,1\}^n$, the \emph{constellation} of $x$ is the set of
  points $\{ [i,x \oplus r_i] \mid 0 < i \leq \ell\}$; that is, the
  set includes, for each $i > 0$, the point in the domain of (the
  implicitly defined) $h_i$ defining $h_i(x \oplus r_i)$. While this
  depends on $R$, we suppress this as it can always be inferred from
  context.
\end{definition}

\subsection{The simulator}
\label{sec:simulator}
The major task of the simulator is to ensure the answers to appropriate $h_0$ queries are consistent with the value of $F(\cdot)$.
In general, queries to the $h_i$, for $i > 0$, are simply treated as
random functions and lazily evaluated and cached as usual. The
function $h_0$ must be given special treatment: ideally, the simulator
would like to ensure the equality
$\tilde{h_0}(\tilde{g}_R(x)) = F(x)$. This poses some challenges
because the simulator cannot typically identify---for a particular $y$
in question---an $x$ for which $y = \tilde{g}_R(x)$; thus there will
be circumstances where the simulator simply cannot correctly program
$h_0(y)$ with the corresponding $F(\cdot)$ value. A further challenge
is that values of $\tilde{h}_0(\cdot)$ cannot, in general, be
``programmed'' to agree with $F$ in circumstances where $\tilde{h_0}$
has been subverted. To minimize the chance that one of the simulator's
``inconsistencies'' is observable by the distinguisher, the simulator
aggressively computes $\tilde{g}_R(x)$ for any $x$ whose constellation
$\{ h_i(x \oplus r_i)\}$ the distinguisher (even partially)
examines. In particular, if the distinguisher queries
$h_i(x \oplus r_i)$, the simulator will immediately determine the
values for all $\tilde{h}_j(x \oplus r_j)$ and, if possible, correctly
program $h_0(\cdot)$ at $\tilde{g}_R(x)$. In this way, the simulator
typically guarantees that constellations which the distinguisher
has examined are correctly programmed. Of course, the distinguisher
may speculatively query $h_0$ on other locations but, with this
simulator, it will not be easy for the distinguisher to force a lie at
a location that will result in a future inconsistency (so long as the
game only runs for polynomially many steps). A similar difficulty
arises in cases where the distinguisher queries an $h_i(x \oplus r_i)$
in a constellation for which $\tilde{h}_0(\tilde{g}_R(x))$ has been
subverted; again, this situation may interfere with the simulator's
ability to properly program $\tilde{h}_0$. Likewise, we will see that
such queries are difficult for the distinguisher to discover.


Formally, we define the simulator $\cS$ (for the abbreviated model)
below. After the proof that this simulator achieves (abbreviated)
crooked indifferentiability (Definition~\ref{def:indiff-crooked}), we
show that the simulator can be lifted to one that achieves full
indifferentiability (Definition~\ref{def:abbrev-indiff-crooked}).

\begin{mdframed}

\begin{center}
  \textsc{The Abbreviated Simulator}
\end{center}

The simulator is provided $R$, the randomness associated with the
security game, and $\langle\tilde{H}\rangle$, the subversion
algorithm. As in our previous discussions, we implicitly associate $H$
with the family of functions $h_i$ for $i \geq 0$ with particular
conventions for padding or embedding inputs and outputs to realize the
variety of domains and ranges for the $h_i$; for simplicity we assume
that all queries (made by $\cD$ or $\tilde{H}$) to $H$ have a unique interpretation
as a query to an $h_i$. Roughly, $\cS$ ``lazily'' maintains the random
oracles via tables that are populated according to conventions that
minimize the risk of a conflict arising from its handling of the
function $h_0$.

\begin{center}
\begin{tabular}{|c|c|}
  \hline
  \multicolumn{2}{|c|}{$h_0$} \\ \hline
  Query  & $h_0(x_i)$  \\ \hline
  $x_1$  & $v_{1,0}$  \\ \hline
  $\vdots$ & $\vdots$  \\ \hline
  $x_{q_0}$ & $v_{q_0,0}$  \\ \hline
\end{tabular}
\begin{tabular}{|c|c|}
  \hline
  \multicolumn{2}{|c|}{$h_1$} \\ \hline
  Query  & $h_1(x_i)$  \\ \hline
  $x_1$  & $v_{1,1}$  \\ \hline
  $\vdots$ & $\vdots$  \\ \hline
  $x_{q_1}$ & $v_{q_1,1}$  \\ \hline
\end{tabular}\quad\ldots\quad
\begin{tabular}{|c|c|}
  \hline
  \multicolumn{2}{|c|}{$h_\ell$} \\ \hline
  Query  & $h_\ell(x_i)$  \\ \hline
  $x_1$  & $v_{1,\ell}$  \\ \hline
  $\vdots$ & $\vdots$  \\ \hline
  $x_{q_\ell}$ & $v_{q_\ell,\ell}$  \\ \hline
\end{tabular}
\captionof{table}{\label{table:query} Tables maintained by $\cS$.}
\end{center}

$\cS$ responds to queries from $\cD$ with following procedure, which
also describes the mechanism for maintaining the value tables.
\begin{itemize}
\item All tables are initially empty.
\item A query $x$ to $h_i()$ is answered as follows. If $i=0$ and
  the table for $h_0$ holds a value for $x$, this is returned and no
  further action is taken. If $x$ does not appear in the table, a
  uniformly random value $v$ is drawn from $\{0,1\}^n$, added to the
  table for $h_0$, and returned.  Otherwise, $i>0$ and the situation
  is more complicated, as $\cS$ wishes to compute the value
  $\tilde{g}$ associated with this query. In this case, the simulator
  then proceeds as follows.
  \begin{enumerate}
  \item \label{step:constellation} $\cS$ determines the ``adjusted''
    query $x' := x \oplus r_{i}$, and ensures that all points in the
    constellation---that is $h_j(x' \oplus r_j)$ for $j > 0$---have
    assigned values: specifically, for each $j$, if the result of
    $h_j(x' \oplus r_j)$ has not yet been determined, this value is
    drawn uniformly and added to the table.
  \item The next task is to evaluate $\tilde{g}_R(x')$, the
    corresponding query to $h_0$, which would ideally be programmed to
    agree with $F$. For this purpose, $\cS$ then runs the
    implementation $\tilde{h}_j$ (using $\tilde{H}$) on
    $x' \oplus r_j = x \oplus r_{j}\oplus r_i$, for all
    $j \in [\ell]$, to derive the value
    \[ \tilde{g}_R(x')=\bigoplus_{j=1}^\ell
      \tilde{h}_j({x}'\oplus r_j),
    \]
    where $[\ell]$ is used to denote the set $\{1,...,\ell\}$.
    
    During the execution of $\tilde{H}$ on those inputs, $\cS$ must
    determine any additional random oracle queries issued by the
    implementation $\tilde{H}^{h_*}$. These are answered by a less
    aggressive convention: $\cS$ first checks whether the query, say
    $h_k(y)$, appears in the table; if so, the stored valued is
    used; if not, $\cS$ uses a uniformly random value $u$ as
    answer and records it in the table. (Note that this does not lead
    to forced evaluation of the full constellations in which these
    queries lie.)
  \item $\cS$ checks whether $\tilde{g}_R(x')$ has been previously
    queried in $h_0$; if not, $\cS$ properly programs $h_0$:
    specifically, $\cS$ queries $F$ on $x'$, collects the response
    $F(x')$, and assigns $h_0(x') = F(x')$. Finally, $\cS$ returns to
    the distinguisher the value appearing in the appropriate table to
    the original query $x$. Note that a value always appears at this
    point in light of Step~\eqref{step:constellation} above.  (N.b.,
    in many cases, for convenience, we assume that $\cS$ returns to
    the distinguisher all values in the constellation of $x$---that
    is, all $h_j(x' \oplus r_i)$ for $j > 0$.)
\end{enumerate}
\end{itemize}
\end{mdframed}

\paragraph{\bf Probability analysis.} We prove indifferentiability by
introducing a sequences of four games that connect the real
construction and the simulator. In Game~\ref{Game:construction-clean}, we define three
crisis events $\subv$, $\pred$ and $\selfref$, which play a central
role in distinguishing Game \ref{Game:construction}, the real construction, from
Game \ref{Game 4}, the simulator interaction game. According to
Theorem \ref{thm:indiff}, the gap between $(\rC, H)$ and $(\cF, \cS)$
is bounded by $\Pr[\subv]+\Pr[\pred]+\Pr[\selfref]$.

Lastly, the three events are shown to be negligible in Theorem \ref{thm:selfref}, Theorem \ref{thm:unpred.}, and Therorem \ref{thm:subv}.

\subsection{The hybrid games; the crisis events}

\subsubsection{The games}

We introduce the security games used to organize the main
proof. Game~\ref{Game:construction} corresponds to interaction with the
construction, while Game~\ref{Game 4} corresponds to interaction with
the simulator described above. The intermediate games correspond to
hybrid settings that are convenient in the proof. The main ``crisis
events'' that reflect circumstances where the simulator could fail to
provide consistency are defined with respect to Game~\ref{Game:construction-clean}.

The description of each game indicates how queries to the two oracles
with which $\cD$ interacts are answered. We adopt the notational
convention that $\mathcal{A}$ denotes the rule for answering oracle
queries to $F$---corresponding to the construction and the ideal
functionality---while $\mathcal{B}$ denotes the rule for answering
oracle queries to $h_i$---corresponding to $H$ and the simulator. For expository
purposes, several variants are presented for each game; these variants
all yield precisely the same transcript with the distinguisher but
reflect differing internal conventions for generating the responses.

\begin{game}[The construction]\label{Game:construction}  $\mathcal{A}=\mathcal{C}^H,\mathcal{B}=H$. To most
  easily compare with the following games, we explicitly express the
  game as a ``data preparation step'' and an ``interaction'' step:
  \begin{enumerate}
  \item Data preparation: Uniformly select the value $h_i(x)$ for each
    $0 \leq i \leq \ell$ and each $x$ in the domain.
  \item Interaction: If $\cD$ queries $\mathcal{A}$ for $x$,
    $\mathcal{A}$ returns {\color{red} $\tilde{h}_{0}(\tilde{g}_R(x))$}; if $\cD$
    queries $\mathcal{B}$ for $(i,x)$, $\mathcal{B}$ returns
    $h_{i}(x)$.
  \end{enumerate}
\end{game}

\noindent\emph{Some conventions and terminology}. In general, we say that ``the distinguisher $\cD$ queries the constellation of 
$x$'' (for $x \in \{0,1\}^n$) when $\cD$ queries any
$h_i(x \oplus r_i)$ term for $i>0$. (This is consistent with the
definition of constellation above.) Beginning with
Game~\ref{Game:construction-clean}.2, we consider procedures for
maintaining function values by lazy evaluation. In this setting we
define the notion of \emph{completing} the constellation of $x$ as the
following procedure: $h_i(x \oplus r_i)$ is evaluated for all elements
in the constellation---that is, they are assigned to uniformly random
values if they have not yet been assigned---followed by evaluation of
$\tilde{h}_i(x \oplus r_i)$ (using the subversion algorithm) for each
$x \oplus r_i$ in the constellation. The evaluation of
$\tilde{h}_i(x_i \oplus r_i)$ will in general require evaluation of
$h_i()$ at a collection of further points, which are likewise assigned
to new uniformly random values unless they have already been
assigned. We also use the notation $\tilde{g}_R(x)$ to denote the sum
of $\tilde{h}_j({x}\oplus r_j)$ for $0<j \leq \ell$. This is directly
analogous to the ``aggressive'' constellation evaluation discussed in
the formal discussion of the simulator above.

The notation $\tilde{g}_R(\cdot)$ described above deserves special
consideration. In Game~\ref{Game:construction}.1, $\tilde{g}_R(x)$ can
be unambiguously defined for all $x$ (as $\tilde{h}_i()$ can be
unambiguously defined); the same can be said of
Game~\ref{Game:construction-clean}.1, below. For subsequent games,
there is no immediate way to define $\tilde{g}_R(x)$ for all $x$, as
these values may depend on the queries made by $\cD$; however, the
syntactic definition of $\tilde{g}_R(x)$ (as the sum of appropriate
$\tilde{h}_j({x}\oplus r_j)$) is well-defined for all constellations
that have been completed by $\cS$; thus, we continue to use this
notation to indicate this sum, with the understanding that it is only
meaningful in settings where the constellation has been completed.

For convenience, we assume the distinguisher does not make repetitive
queries.
\begin{game}[The construction with unsubverted $h_0$]
  \label{Game:construction-clean} Game~\ref{Game:construction-clean} has two variants. Game~\ref{Game:construction-clean}.1 is identical to Game~\ref{Game:construction}
    with the exception that the construction is replaced with an
    idealized one that uses the unsubverted $h_0$. Game~\ref{Game:construction-clean}.2 merely shifts the perspective of Game~\ref{Game:construction-clean}.1 to one involving tables.
  \begin{itemize}
  \item\textbf{Game~\ref{Game:construction-clean}.1:}
    \begin{enumerate}
    \item Data preparation: Uniformly select the value $h_i(x)$ for each $0 \leq i \leq \ell$ and each $x$ in the domain.
    \item Interaction: If $\cD$ queries $\mathcal{A}$ for $x$,
      $\mathcal{A}$ returns {\color{red} $h_{0}(\tilde{g}_R(x))$}; if $\cD$ queries
      $\mathcal{B}$ for $(i,x)$, $\mathcal{B}$ returns $h_{i}(x)$.
    \end{enumerate}
  \item\textbf{Game~\ref{Game:construction-clean}.2:} $\mathcal{A}$ and $\mathcal{B}$ maintain tables,
    $T_F$ and $T_H$, respectively, storing the point evaluations of
    $F$ and $H$ arising during the interaction. Initially, the tables
    are empty. Interaction with $\cD$ proceeds as follows:
      \begin{itemize}      
      \item[-] If $\cD$ queries $\mathcal{B}$ for $(i,x)$: if
        $h_{i}(x)$ appears in the table $T_H$, return the recorded
        value; otherwise, draw a uniform value, assign $h_{i}(x)$ to
        the value by adding this entry to the table $T_H$, and return
        it. Additionally, if $i>0$, complete the constellation of
        $x \oplus r_i$ and uniformly assign
        $h_0(\tilde{g}_R(x \oplus r_i))$ in $T_H$ unless it has been
        previously assigned. {\color{red}If $F(x)$ is unassigned in $T_F$, it is
        assigned to the value $h_0(\tilde{g}_R(x \oplus r_i))$.} Otherwise $F(x)$ was already
      assigned and this conflict is not explicitly managed.
      \item[-] If $\cD$ queries $\mathcal{A}$ for $x$: if $F(x)$
        appears in the table $T_F$, return the recorded value;
        otherwise, complete the constellation of $x$ and {\color{red}properly
        program $F(x)$ at the value $h_0(\tilde{g}_R(x))$}. Specifically,
        completion determines the value $\tilde{g}_R(x)$; if
        $h_0(\tilde{g}_R(x))$ is unassigned, it is assigned a random
        value in $T_H$. Finally, $\mathcal{A}$ returns
        $h_{0}(\tilde{g}_R(x))$ and $F(x)$ is assigned to this value
        in $T_F$.
      \end{itemize}
    \end{itemize}  
\end{game}

For technical reasons later, we fill in the empty cells of $T_F$ and $T_H$ in Game \ref{Game:construction-clean}.2 after the interaction. Empty cells of $T_H$ are filled in by uniformly chosen values while empty cells of $T_F$ are filled in by $F(x):=h_{0}(\tilde{g}_R(x))$. The resulting $T_H$ is denoted by $h_*$. Also, in Game \ref{Game:construction-clean}.2, we say $y \in \{0,1\}^{3n}$ is in the \emph{subverted area} of $h_0$ if, in $h_*$, $\tilde{h}_0(y) \neq h_0(y)$.


\begin{game}\label{Game 3} Game~\ref{Game 3} also has two variants. Game~\ref{Game 3}.1 is a small pivot from Game~\ref{Game:construction-clean}.2 which
  handles conflicts by a different convention. Game~\ref{Game 3}.2 is identical
    to Game~\ref{Game 3}.1 with all randomness selected prior to the game.
  \begin{itemize}
  \item \textbf{Game~\ref{Game 3}.1:} $\mathcal{A}$ and $\mathcal{B}$ maintain tables, $T_F$ and $T_H$ respectively, storing the evaluated terms during the interaction. Initially, the data sets are empty. Interaction with $\cD$ proceeds as
    follows:
    \begin{itemize}
    \item[-] If $\cD$ queries $\mathcal{B}$ for $(i,x)$: if
        $h_{i}(x)$ appears in the table $T_H$, return the recorded
        value; otherwise, draw a uniform value, assign $h_{i}(x)$ to
        the value by adding this entry to the table $T_H$, and return
        it. Additionally, if $i>0$, complete the constellation of
        $x \oplus r_i$ and uniformly assign
        $F(x)$ in $T_F$ unless it has been
        previously assigned. {\color{red}If $h_0(\tilde{g}_R(x \oplus r_i))$ is
      unassigned in $T_H$, it is assigned to the value $F(x)$} and we say that
      $h_0$ has been \emph{programmed}.  Otherwise $h_0$ was already
      assigned and this conflict is not explicitly managed.
    \item[-] If $\cD$ queries $\mathcal{A}$ for $x$: if $F(x)$
        appears in the table $T_F$, return the recorded value;
        otherwise,uniformly
      select assign $F(x)$ and return this
      value, complete the constellation of $x$ and {\color{red}properly
        program $h_0(\tilde{g}_R(x))$ at the value $F(x)$}. Specifically,
        completion determines the value $\tilde{g}_R(x)$; if
        $h_0(\tilde{g}_R(x))$ is unassigned, it is assigned to $F(x)$ and we say that $h_0$ has been
      \emph{programmed}; otherwise $h_0$ was already assigned and this
      conflict is not explicitly managed.
    \end{itemize}
  \item \textbf{Game~\ref{Game 3}.2:} Game~\ref{Game 3}.2 is {\color{red} identical
    to Game~\ref{Game 3}.1, but all randomness is selected \emph{a
      priori}}; note that responses are still defined in terms of
    tables as defined in 3.1. Two data sets, $\DS_1$ and $\DS_2$, are
    drawn: $\DS_1$ contains uniformly selected values for each $F(x)$
    and $h_i(x)$ for all $0 < i \leq \ell$ and all $x \in \{0,1\}^n$;
    $\DS_2$ contains uniformly selected values $h_0(y)$ for all
    $y \in \{0,1\}^{3n}$. (Note that all entries are selected
    independently and uniformly.) The game is a straightforward
    adaptation of Game~\ref{Game 3}.1: When $\mathcal{A}$ or
    $\mathcal{B}$ generate values (for table insertion) for $T_F$ or
    $T_H$ (for $h_i$ for $i > 0$), these values are drawn from
    $\DS_1$. The function $h_0$ is treated specially: If $h_0$ is
    programmed at $y$, the value $h_0(y)$ is drawn from (the table for
    $F$ in) $\DS_1$. Otherwise the value is drawn from
    $\DS_2$. 
  \end{itemize}
  \end{game}
  
  \begin{game}[The abbreviated simulator] \label{Game 4} Observe that
    queries to $\mathcal{A}$ (that is $\mathcal{F}$) in the interaction between
    $\cD$ and $(\mathcal{F}, \mathcal{S}^{\mathcal{F}})$ do not result
    in associated queries to $\mathcal{S}$ (cf.\ Game~\ref{Game 3}). Game~\ref{Game 4} is
      {\color{red}identical to the interaction between $\cD$ and $(\mathcal{F}, \mathcal{S}^{\mathcal{F}})$, but all randomness is selected
      \emph{a priori}}. Prepare two data sets $\DS_1$ and $\DS_2$ :
      $\DS_1$ contains uniformly selected values for $F(x)$ and
      $h_i(x)$ for all $0<i \leq \ell$ and all $x \in \{0,1\}^n$;
      $DS_2$ contains uniformly selected $h_0(y)$ for all
      $y \in \{0,1\}^{3n}$. (Note that all data are selected
      independently and uniformly.) Game~\ref{Game 4}.1 is then
      carried out with the following conventions for randomness: All
      queries to $F$ (either direct responses from $\mathcal{A}$ or
      queries by $\mathcal{B} = \mathcal{S}$, the simulator) are
      determined by $\DS_1$. When $\mathcal{B}$ evaluates $h_i$ for
      $i > 0$ and records these assignments in $T_H$, values are
      likewise drawn from $\DS_1$.  When $\mathcal{B}$ assigns
      $h_0(y)$ in $T_H$ to a ``fresh random value,'' this is drawn from
      $\DS_2$; when $h_0(y)$ is programmed in $T_H$ to coincide with
      $F$, this value is drawn from $\DS_1$. 
\end{game}

We define the following \emph{crisis events} with respect to a
distinguisher $\cD$ in Game~\ref{Game:construction-clean}.2.
\begin{align*}  
  \pred & =\left\{\parbox{13cm}{for some $x$, the distinguisher queries the constellation of $x$ for the first time and $h_0(y)$, where $y = \tilde{g}_R(x)$, has already been assigned a value by the simulator}\right\}\,,\\
  \subv &=\{\text{the distinguisher queries a constellation $x$ such that $\tilde{g}_R(x)$ falls into the subverted area of $h_0$} \}\,,\\
  \selfref &=\left\{\parbox{13cm}{for some $x$, the distinguisher queries a
             constellation $x$ such that $h_0$ at $y = \tilde{g}_R(x)$ is
             evaluated during completion of the constellation $x$}
             \right\}\,.
\end{align*}
Note that the crisis event $\pred$ involves the \emph{first} query
made to the constellation of $x$. (This distinction can be ignored
under the assumption that the simulator returns all constellation
values to the distinguisher which are treated as regular distinguisher
queries and so will not be repeated.)

\begin{definition}[Transcript of a game]
  For one of the games above, a distinguisher $\cD$ and randomness
  $R$, we define the \textbf{$R$-transcript} (or \textbf{transcript})
  of the game as the random variable given by $R$ and the ordered
  sequence of (query, answer) pairs. We encode (query, answer) pairs
  as tuples of the form $(x,F,y)$ or $(x,h_i,y)$, which indicates that
  $F(x)=y$ and $h_i(x)=y$, respectively. For the transcript $\alpha$,
  we let $\alpha[k]$ denote the first $k$ pairs in $\alpha$. As a
  reminder that a transcript $\alpha$ determines $R$, we let
  $\alpha[0]=R$.
\end{definition}

The notational convention that $\alpha[0] = R$ is occasionally useful
in induction proofs of equality between two transcripts---equality at
zero indicates that they correspond to the same randomness.

The analysis focuses on the transcripts of a fixed deterministic
distinguisher arising from its interactions in Game
\ref{Game:construction} through \ref{Game 4}.  As mentioned above, and
discussed more fully below, the variants of each game (e.g., Games
\ref{Game:construction-clean}.1 and \ref{Game:construction-clean}.2) yield precisely the same
distributions of transcripts. (Indeed, for the same random choices
these games yield identical transcripts.) To reflect this, for a fixed
deterministic distinguisher we adopt the notation $\alpha_i$, for
$i=1,2,3,4$, to denote the transcript arising from interaction given
by Game $i$.

Without loss of generality, we assume at the end of the interaction
there exists no $x \in \{0,1\}^n$ such that $F(x)$ is queried but the
constellation of $x$ is not. Observe that any distinguisher may be
placed in this normal form by appending any necessary extra queries at
the end of its native sequence of queries to appropriately query every
constellation associated with a query to $F$.

\begin{theorem}\label{thm:indiff} For any distinguisher $\cD$, 
  $\displaystyle\|\alpha_1-\alpha_4\|_{\tv} \leq
  \Pr[\subv]+\Pr[\pred]+\Pr[\selfref]$.
\end{theorem}

Theorem \ref{thm:indiff} is proved by the sequence of lemmas below.
We recall, at the outset, that if $X: \Omega \rightarrow V$ and
$Y: \Omega \rightarrow V$ are two random variables defined on the same
probability space $\Omega$ and $\Pr[X \neq Y] = \epsilon$, then
$\| X - Y \|_{\tv} \leq \epsilon$.

\begin{lemma}\label{lemma:equivalent}
  The transcripts of Game~\ref{Game:construction-clean}.1 and Game~\ref{Game:construction-clean}.2 have the same
  distribution.
\end{lemma}
\begin{proof}
  It suffices to notice that the distributions of the query in Game~\ref{Game:construction-clean}.1 and Game~\ref{Game:construction-clean}.2 are identical because the rule of Game~\ref{Game:construction-clean}.2
  guarantees all the $h_i(\alpha)$ for $0 \leq i \leq \ell$, like
  those in Game~\ref{Game:construction-clean}.1, are chosen uniformly and independently.
\end{proof}

Notice that the three crisis events defined in Game \ref{Game:construction-clean}.2 can also be properly defined in \ref{Game:construction-clean}.1. The proof of Lemma \ref{lemma:equivalent} reveals that their probabilities are the same for Game \ref{Game:construction-clean}.1 and \ref{Game:construction-clean}.2 since both of the two versions of Game \ref{Game:construction-clean} have the same distribution of $h_*$. In the rest of the paper, we abuse the notation and use $\subv$, $\pred$ and $\selfref$ to denote the crisis events in all the three versions of Game \ref{Game:construction-clean}.

\begin{lemma}\label{lemma: 1vs2} For any distinguisher $\cD$, 
  $\|\alpha_1-\alpha_2\|_{\tv}\leq \Pr[\subv]$.
\end{lemma}

\begin{proof}
Notice that if the data preparation step prepares the same $h_*$ for Game \ref{Game:construction} and Game~\ref{Game:construction-clean}.1, the only way to distinguish them is to find a $\tilde{g}_R(x)$ falling into the subverted area of $h_0$. That is to say, the transcripts in the two games are different only when $\subv$ occurs.
\end{proof}

\begin{lemma}\label{lemma: 2vs3} For any distinguisher $\cD$, 
  $\|\alpha_2-\alpha_3\|_{\tv}\leq \Pr[\pred]+\Pr[\selfref]$.
\end{lemma}
\begin{proof}
  Imagine the two parties $(\mathcal{A},\mathcal{B})$ use the same random bits in Game~\ref{Game:construction-clean}.2 and Game~\ref{Game 3}.1. Then, Game~\ref{Game:construction-clean}.2 and Game~\ref{Game 3}.1 have the same query result unless $\pred$ or $\selfref$ occurs in Game~\ref{Game:construction-clean}.2, which has the probability $\Pr[\pred]$.
\end{proof}

\begin{lemma}
  For any distinguisher $\cD$ the transcripts of Game~\ref{Game 3}.1 and Game~\ref{Game 3}.2
  have the same distribution.
\end{lemma}
\begin{proof}
  The distributions of the entrees of $T_F$ and $T_H$ in Game~\ref{Game 3}.1 are same as those in
  Game~\ref{Game 3}.2, these games have the same transcript
  distributions.
\end{proof}

In the next three lemmas, we denote by $\beta_3[k]$ (and $\beta_4[k]$,
respectively) the data in $T_H$ in Game~\ref{Game 3}.2 (and Game~\ref{Game 4},
respectively) after $k$th query and answer. Similar to the conventions in $\alpha$, we let $\beta_3[0]$($\beta_4[0]$) be $R$. Elements in $\beta$,
$\DS_1$, and $\DS_2$ are given the same tuple syntax as those in the
transcript. For any positive $k$, we say $\beta_3[k]$
\emph{agrees with} $\beta_4[k]$
if, for any $m \leq k$, $(x,h_0,a) \in \beta_3[m]$ and
$(x,h_0,b) \in \beta_4[m]$, we have $a=b$. For
$(x,h_0,a) \in \beta_3[k]$ (or $\beta_4[k]$), we say the pair
$(x,h_0)$ is a \emph{free term} in $\beta_3[k]$ (or $\beta_4[k]$) if
there is a $(x,h_0,b) \in \DS_2$ such that $a=b$. If the term
$(x,h_0)$ is not free, we say that it is \emph{programmed}. (These
notions only apply to $h_0$, reflecting the two mechanisms for
assigning $h_0$ values in the games above.) 

\begin{lemma}\label{lemma:agreement}
  Consider the interaction of a distinguisher $\cD$ with
  Games~\ref{Game 3}.2 and~\ref{Game 4} given by the same randomness
  $\DS_1$, $\DS_2$ and $R$.  For any $k>0$, if $\beta_3[k]$ agrees with
  $\beta_4[k]$ then $\alpha_3[k] = \alpha_4[k]$ and the $k+1$st
  queries of the two games are identical.
\end{lemma}

\noindent
Remark: Note that agreement, in the sense defined above, is much
weaker than equality or setwise equality. However,
Lemma~\ref{lemma:agreement} implies that to prove $\alpha_3=\alpha_4$,
it suffices to show that $\beta_3[k]$ agrees with $\beta_4[k]$ for any $k$.

\begin{proof}
  We will prove by induction that $\alpha_3[i] = \alpha_4[i]$ for each
  $0 \leq i \leq k$.  Of course $\alpha_3[0] = \alpha_4[0]$, as these
  are both $R$.  Suppose that
  $\alpha_3[i-1] = \alpha_4[i-1]$ for an $1 \leq i \leq k$. Since the
  first $i-1$ queries have been identical and have received the same
  responses, the $i$th query in Game~\ref{Game 3}.2 and Game~\ref{Game
    4} are identical. 
    If the $i$th query is made to $F$ or $h_i$ for
  $i > 0$, the responses are identical because the datasets are
  equal. Otherwise, the query is made to $h_0$, in which case the
  responses are identical because $\beta_3[k] = \beta_4[k]$. Of course
  this also implies that the $k+1$st queries are identical.
\end{proof}

\begin{lemma}\label{lemma:subset}
  Consider the interaction of a distinguisher $\cD$ with
  Games~\ref{Game 3}.2 and~\ref{Game 4} given by the same randomness
  $\DS_1$, $\DS_2$ and $R$.  For any $k>0$, if $\beta_3[k]$ agrees with
  $\beta_4[k]$, $\beta_4[k] \subset \beta_3[k]$.
\end{lemma}

\begin{proof}
  Assume that $\beta_3[k]$ agrees with $\beta_4[k]$; we proceed by
  induction to show that $\beta_4[j] \subset \beta_3[j]$ for all
  $j \leq k$.  Of course $\beta_4[0] \subset \beta_3[0]$ as these are
  both $R$. For an $j$ in the range $1\leq j \leq k$,
  suppose $\beta_4[j-1] \subset \beta_3[j-1]$. Note that the $j$th
  query and response in Game~\ref{Game 3}.2 and Game~\ref{Game 4} are identical by
  Lemma~\ref{lemma:agreement}. We consider the various cases
  separately:
\begin{itemize}
\item If the $j$th query is $(x,F)$ for some $x$: Then
  $\beta_4[j] = \beta_4[j-1] \subset \beta_3[j-1] \subset \beta_3[j]$,
  as desired. Note, in general, that $\beta_3[j-1] \neq \beta_3[j]$
  for such a query.
\item If the $j$th query is $(x,h_0)$ for some $x$: Note that
  $\beta_4[j]/\beta_4[j-1] \subset \{(x,h_0,y)\}$ for some $y$ and
  that $(x,h_0,y) \in \beta_3[j]$.
\item If the $j$th query is $(x,h_i)$ for some $x$ and positive $i$:
  In both games, this calls for the constellation to be completed and
  $h_0(\tilde{g}_R(x \oplus r_i))$ to be evaluated.
  As the datasets are equal in the two games, the only possible
  disagreement that could arise during completion of the
  constellations must occur on a query to $h_0$: however, this is
  excluded by agreement of $\beta_3[k]$ and $\beta_4[k]$. It follows
  that the $h_i$ are evaluated at exactly the same sequence of inputs,
  yielding the same results, by the two completions and hence that any
  table assignment added to $\beta_4$ must also appear in $\beta_3$;
  additionally, $\tilde{g}(x \oplus r_i)$ takes the same value in the
  two games. Finally, $h_0(\tilde{g}(x \oplus r_i))$ is evaluated in
  both games which, by the same considerations, must yield the same
  value and must be included in both $\beta_3$ and $\beta_4$.
  \qedhere
\end{itemize}
\end{proof}

Next, we introduce two further events in Game~\ref{Game 3}.2 and proceed to our
main result.

\begin{itemize}
\item \textbf{Forward Prediction}, denoted $\forwardpred$: For some
  $y$, the constellation of $y$ is completed for the first time at
  which point it is discovered that $h_0(\tilde{g}(y))$ has previously
  been evaluated.
\item \textbf{Backward Prediction}, denoted $\backwardpred$: For some
  $x$, $h_0(\tilde{g}_R(x))$ is queried, either directly by $\cD$ or
  during completion, after $\cD$ has already queried $F(x)$ (to
  $\mathcal{A}$) but prior to $\cD$ making any query to the
  constellation $x$.
\end{itemize}

\begin{lemma}\label{PP for 3v4}
Consider the interaction of a distinguisher $\cD$ with Games~\ref{Game 3}.2 and~\ref{Game 4} given by the same randomness $\DS_1$, $\DS_2$ and $R$. For any $k>0$, if $\forwardpred$ and $\backwardpred$ do not occur in Game~\ref{Game 3}.2, $\beta_3[k]$ agrees with $\beta_4[k]$.
\end{lemma}
\begin{proof}
  We proceed by induction. First, $\beta_3[0]$ agrees with
  $\beta_4[0]$. Suppose now that $\beta_4[k-1]$ agrees with
  $\beta_3[k-1]$: we analyze three cases to ensure that $\beta_4[k]$
  agrees with $\beta_3[k]$. Note, at the outset, that from
  Lemma~\ref{lemma:agreement} the $k$th queries made by $\cD$ are
  identical in the two games.
 \begin{itemize}
 \item[-] If $\cD$ queries $(x,F)$ for some $x$:
   $\beta_4[k]=\beta_4[k-1]$ and hence agrees with $\beta_3[k]$.
  \item[-]  If $\cD$ queries for $(x,h_0)$ for some $x$:
  \begin{enumerate}
  \item if $(x,h_0)$ is not evaluated yet in $\beta_3[k-1]$, by Lemma
    2, it is not in $\beta_4[k-1]$, either. Therefore,
    $\beta_4[k]/\beta_4[k-1]=\beta_3[k]/\beta_3[k-1]=(x,h_0,a)$ for
    some $a$ (given by $\DS_2$). $\beta_4[k]$ and $\beta_3[k]$ agree.
  \item if $(x,h_0)$ is evaluated in both $\beta_3[k-1]$ and
    $\beta_4[k-1]$. $\beta_4[k]=\beta_4[k-1]$ agrees with
    $\beta_3[k]=\beta_3[k-1]$.
  \item if $(x,h_0)$ is evaluated in $\beta_3[k-1]$ but not in
    $\beta_4[k-1]$. We have two subcases here. First, when $(x,h_0)$
    is free in $\beta_3[k-1]$: In this case it will be free in
    $\beta_4[k]$ and hence $\beta_4[k]$ agrees with
    $\beta_3[k]$. Second, $(x,h_0)$ is programmed in $\beta_3[k-1]$:
    We show this is actually impossible. Notice that no constellation
    $y$ for which $\tilde{g}(y) = x$ has been queried by $\cD$ in
    Game~\ref{Game 4} since $(x,h_0)$ is not in
    $\beta_4[k-1]$. Since $\beta_4[k-1]$ agrees with $\beta_3[k-1]$,
    no such constellation has been queried in Game~\ref{Game 3}.2
    either. Therefore, in this case the event $\backwardpred$ would
    have occurred. 
  \end{enumerate}
\item[-] If $\cD$ queries a constellation $x$: In both games, this
  leads to completion of the constellation of $x$, including
  evaluation of the last term $h_0(\tilde{g}_R(x))$.
  \begin{enumerate}
  \item We first analyze the behaviour of the evaluations prior to the
    last term. We need to show the sequence of evaluations in
    Game~\ref{Game 3}.2 coincides with that in Game~\ref{Game
      4}.
    No possible disagreement can arise on any term that is free in
    both games, as these are both drawn from $\DS_1$; thus any
    disparity must occur at a term programmed in at least one game. If
    a term is programmed only in Game~\ref{Game 3}.2, the event
    $\backwardpred$ occurs, contrary to assumption. Observe that if a
    term is programmed in Game~\ref{Game 4} it must also be in
    Game~\ref{Game 3}.2, as $\beta_4[k-1] \subset \beta_3[k-1]$ from
    Lemma~\ref{lemma:subset} and these agree. Finally, if a term is
    programmed in both games they are given the same value since
    $\beta_4[k-1]$ agrees with $\beta_3[k-1]$.
  \item Now we proceed to the last term $h_0(\tilde{g}(x))$. If it is
    not in $\beta_3[k-1]$, then it is not in $\beta_4[k-1]$ by
    Lemma~\ref{lemma:subset}. Thus, the value of the term is either
    programmed in both games to be consistent with $F$ or was in fact
    freely assigned in both games during completion of the
    constellation itself. In either case, the values in two games are
    equal. If the last term is in $\beta_3[k-1]$, the event
    $\forwardpred$ occurs.
  \end{enumerate}
 \end{itemize}
In summary,  
 $\beta_3[k]$ agrees with $\beta_4[k]$.
\end{proof}

\begin{lemma}\label{lemma: 3vs4}
$\|\alpha_3-\alpha_4\|_{\tv}\leq \Pr[\backwardpred]+\Pr[\forwardpred]<2\Pr[\pred]$, where  $\backwardpred$ and $\forwardpred$ are events in Game \ref{Game 3} and $\pred$ is in Game \ref{Game:construction-clean}.
\end{lemma}
\begin{proof}
We view the two pairs of $(\mathcal{A},\mathcal{B})$ in Game \ref{Game:construction-clean}.3 and Game \ref{Game 3}.1 as two probabilistic Turing machines sharing the same randomness $r$ for any long enough string $r$. Given the fixed distinguisher $\cD$ from Lemma \ref{lemma: 1vs2} to \ref{PP for 3v4} and $r$ above, $\pred$ occurs in Game~\ref{Game:construction-clean}.1 if $\backwardpred$ or $\forwardpred$ occurs in Game~\ref{Game 3}.2. Since $r$ and $\cD$ are arbitrary, we have $\Pr[\backwardpred]<\Pr[\pred]$ and $\Pr[\forwardpred]<\Pr[\pred]$, which implies the lemma.
\end{proof}

Theorem~\ref{thm:indiff} follows by applying the triangle inequality to Lemma~\ref{lemma: 1vs2}, Lemma~\ref{lemma: 2vs3}, and Lemma~\ref{lemma: 3vs4}. 

\subsection{Establishing pointwise unpredictability and subversion-freedom}

We first prove that the events $\pred$ and $\subv$ are
negligible. Throughout we let $\epsilon$ denote the upper bound on the disagreement probability of the subversion algorithm $\tilde{H}$: (for all $i$)
$h_i(x) \neq \tilde{h}_i(x)$ with probability no more than $\epsilon$
(in uniform choice of $x$).

\begin{definition}[Good term] Fix a subversion algorithm $\tilde{H}$.
  For any $i \in [\ell]$ and $x \in \{0,1\}^n$, we say $(x,h_i)$ is
  \emph{good} if
  $\Pr_{h_*}[h_i(x) \neq \tilde{h}_i(x)]<\sqrt{\epsilon}$.
\end{definition}

\begin{lemma}
  For any $\tilde{H}$, any $i \in [\ell]$, and uniform $x$,
  $\Pr_{x}[\text{$(x,h_i)$ is good}]>1-\sqrt{\epsilon}$.
\end{lemma}
\begin{proof}
  Notice that
  $\Exp_{x}[\Pr_{h_*}[h_i(x) \neq \tilde{h}_i(x)]]=\Pr_{x, h_*}[h_i(x)
  \neq \tilde{h}_i(x)]<\epsilon$. The lemma therefore follows directly
  from Markov's inequality.
\end{proof}

In several circumstances, is it convenient to determine a hash
function $\resamp h_*$ by ``resampling'' a particular entry of a hash
function $h_*$. To make this precise, for a hash function
$H: \{0,1\}^{3n} \rightarrow \{0,1\}^{3n}$, we define $\resamp[z;s] H$
to be the function given by the rule
\[
  \resamp[z;s] H(x) = \begin{cases}
    H(x) & \text{if $x \neq z$},\\
    s & \text{if $x = z$},
  \end{cases}
\]
We typically apply this when $s$ is drawn uniformly at random in
$\{0,1\}^{3n}$; we then say that the value of $H$ has been
``resampled'' at $z$. In this case, $\resamp[z;s] H$ is a random
variable which we write $\resamp[z] H$ (obtained by selecting $s$
uniformly); when $z$ is explicitly identified by context, we simply
write $\resamp H$. Finally, we apply this notation consistently with
our conventions for rendering the family $h_*$ from $H$. Thus
$\resamp[i,x] h_*$ denotes the family of functions resulting from
resampling $(x,h_i)$.

\begin{fact}[Invariance under resampling]
  Consider an (unbounded) algorithm $A$ with oracle access to
  $H: \{0,1\}^m \rightarrow \{0,1\}^m$, a uniformly random
  function. Suppose, further, that $A^H$ carries out a collection of
  (adaptively chosen) queries to $H$ and with probability 1 returns an
  element $\mathrm{result}(A^H) \in \{0,1\}^m$ on which it has not
  queried $H$. Then the random variable
  $\resamp[\mathrm{result}(A^H)] H$ is uniform; that is, the
  distribution arising from sampling $H$, querying according to $A$,
  and resampling $H$ at $\mathrm{result}(A^H)$ is uniform.
\end{fact}

\begin{proof}
  Observe that the resampling does not change the distribution of $H$
  conditioned on the responses the queries of $A$.
\end{proof}

\begin{definition}[Honest term]
  Fix a subversion algorithm $\tilde{H}$. Then for a fixed $h_*$, an
  index $i \in [\ell]$, and an $x \in \{0,1\}^n$, we say $(x,h_i)$ is
  \emph{honest} if for any $j \in [\ell]$ and $y \in \{0,1\}^n$,
  \[
    \Pr_{s}[\resamp h_i(x) \neq \widetilde{\resamp
      h}_i(x)]<\epsilon^{1/4}\,,
  \]
  where the randomness is over the resampling of $h_j(y)$ (so the
  resampling operator is $\resamp[j,y; s]$ in both cases). A term is \emph{dishonest}
  if it is not honest. (N.b.\ the function(s)
  $\widetilde{\resamp h_i}$ are defined by $\tilde{H}^{\resamp h_*}$ and
  hence may be very different from $h_*$.)
\end{definition}

Let us consider the following random variables defined by (uniform)
random selection of $h_*$:
\[
  d_i(x) = \begin{cases} 1 & \text{if $(x,h_i)$ is dishonest},\\
    0 & \text{otherwise.}
  \end{cases}
\]

\begin{lemma}\label{Honest term is rare.}
  $\displaystyle \Pr_{h_*}[\exists i \in [\ell], \text{$h_i$ has
    more than $2^n(q_{\tilde{H}}+1)\epsilon^{1/8}$ dishonest
    inputs}]\leq\ell\epsilon^{1/8}.$
\end{lemma}
\begin{proof}
If $(x,h_i)$ is good, then
\begin{align*}
  \mathbb{E}[d_i(x)]
  & = \Pr_{h_*}\left[ \exists(j,y), \Pr_{s}\left[\resamp[j,y;s] h_i(x) \neq \widetilde{\resamp[j,y;s] h}_i(x)\right]>\epsilon^{1/4}\right]\\
  & \leq \sum_{k=1,...,q_{\tilde{H}}}\Pr_{h_*}\Bigl[\Pr_{s}\left[\resamp[j,y;s] h_i(x) \neq \widetilde{\resamp[j,y;s] h}_i(x)\right]>\epsilon^{1/4}, \\
  & \mkern25mu  \text{where $h_j(y)$ is the $k$-th term queried by the evaluation of } \tilde{h}_i(x)\Bigr]\\
  &\leq q_{\tilde{H}} \cdot \Pr_{h_*}\Bigl[h_i(x) \neq \tilde{h}_i(x)\Bigr] \leq q_{\tilde{H}}\epsilon^{1/4}
\end{align*}
by the union bound and invariance under resampling. Therefore,
$$\sum_{x \in \{0,1\}^n}\mathbb{E}[d_i(x)]<2^n(q_{\tilde{H}}\epsilon^{1/4}+\sqrt{\epsilon})<2^n(q_{\tilde{H}}+1)\epsilon^{1/4}.$$

Applying Markov's inequality and union bound again, we conclude
\[
  \Pr_{h_*}[\exists i \in [\ell], \text{$h_i$ has more than
    $2^n(q_{\tilde{H}}+1)\epsilon^{1/8}$ dishonest
    inputs}] \leq \ell\epsilon^{1/8}\,.\qedhere
\]
\end{proof}

\begin{definition}[Invisible term]
  Given $(R,h_*)$, we say a term $(x,h_i)$ is \emph{invisible} if for
  any $\tilde{h}_{j}(y)$ that queries $(x,h_i)$, $\tilde{h}_{j}(y)=h_j(y)$ and
 \[
    \Pr_{s}\bigl[\resamp[i,x;s] h_j(y) \neq \widetilde{\resamp[i,x;s]
      h}_j(y)\bigr]<\epsilon^{1/4}\,.
  \]
\end{definition}

\begin{lemma}
  $\Pr[\text{There exists a constellation with fewer than $\ell-n$
    invisible terms}]=O(\ell \epsilon^{1/8}).$
\end{lemma}
\begin{proof}
  We say $h_*$ is stealthy if, for all $i \in [\ell]$, there are no
  more than $\ell2^n(q_{\tilde{H}}+1)\epsilon^{1/8}$ dishonest terms
  in $h_i$. If $h_*$ is stealthy, the dishonest terms in $h_*$ make at
  most $\ell2^nq_{\tilde{H}}(q_{\tilde{H}}+1)\epsilon^{1/8}$ queries
  via the subversion algorithm. In other words, each column contains
  at least $2^n-\ell2^nq_{\tilde{H}}(q_{\tilde{H}}+1)\epsilon^{1/8}$
  terms that are not queried by any dishonest term.

  Notice that a term that is only queried by honest terms is
  invisible. Therefore,
\begin{align*}
    & \mkern20mu \Pr_{R, h_*}[\text{There exists a constellation with more than $n$ visible terms.}]\\
    & < \Pr_{R, h_*}[\text{There exists a constellation with more than $n$ visible terms} \mid \text{$h_*$ is stealthy}]+\Pr_{h_*}[\text{$h_*$ is stealthy}]\\
    & < \Pr_{R, h_*}\left[\parbox{8cm}{There exists a constellation with more than $n$ terms that are queried by a dishonest or subverted term}\;\middle|\;\text{$h_*$ is stealthy}\right] +\ell \epsilon^{1/8}\\
    & <2^n\binom{\ell}{n}(\ell q_{\tilde{H}}(q_{\tilde{H}}+1)\epsilon^{1/8}+\ell q_{\tilde{H}}\epsilon)^n+\ell \epsilon^{1/8}=O(\ell \epsilon^{1/8}),
\end{align*}
as long as $\ell = O(n)$. In the last equality, we use the fact that $q_{\tilde{H}}$ is polynomial in $n$, and for sufficiently large $n$,
\[2^n\binom{\ell}{n}(\ell q_{\tilde{H}}(q_{\tilde{H}}+1)\epsilon^{1/8}+\ell q_{\tilde{H}}\epsilon)^n=O(2^n 4^n \epsilon^{n/16})=O((8\epsilon^{1/16})^n)=O(\epsilon^{1/8})\,.\qedhere
\]
\end{proof}

\begin{definition}[Normal randomness]
  We say $(R, h_*)$ is \emph{normal} if every constellation has at least $\ell-n$ invisible terms. 
\end{definition}

For convenience, we focus on a fixed distinguisher in Game \ref{Game:construction-clean} for the rest of the section.

\begin{definition}[Regular transcript]
  For any $0<k<q_{\hat{\cD}}$, we say an $R$-transcript $\alpha_2[k]$
  is \emph{regular} if
  \[
    \Pr_{h_*}[\text{$(R, h_*)$ is normal} \mid \alpha_2[k]]>1-\sqrt{\ell}\epsilon^{1/16}\,.
    \]
    Otherwise, we say $\alpha_2[k]$ is \emph{irregular}.
\end{definition}

\begin{lemma}
For any $0<k<q_{\hat{\cD}}$.
\[
  \Pr[\text{$\alpha_2[k]$ is irregular}]< \sqrt{\ell}\epsilon^{1/16}\,.
\]
\end{lemma}
\begin{proof}
  \[
    \sum_{\alpha_2[k]}\Pr[\text{$(R, h_*)$ is normal} \mid \alpha_2[k]]\cdot \Pr[\alpha_2[k]]=\Pr[\text{$(R, h_*)$ is normal}]=1-O(\ell \epsilon^{1/8})\,.
    \]
Using Markov's inequality, we have, with probability at least $1-\sqrt{\ell}\epsilon^{1/16}$ in the choice of $\alpha_2[k]$,
\[
  \Pr[\text{$(R, h_*)$ is not normal} \mid \alpha_2[k]]<\sqrt{\ell}\epsilon^{1/16}\,. \qedhere
\]
\end{proof}

\begin{definition}[Consistency]
  For any $0<k<q_{\hat{\cD}}$ and any two functions $h_{*}^1$, $h_{*}^2$, we say $h_{*}^1$ is $k$-consistent with $h_{*}^2$ if for some $R$ in Game \ref{Game:construction-clean}, $h_*=h_{*}^1$ and $h_*=h_{*}^2$ give the same $R$-transcript $\alpha_2[k]$ at the end of $k$-th interaction($k$-th query made by the distinguisher and answer made by the simulator).
\end{definition}

\begin{lemma}\label{lemma:unpre}
For any $0<k<q_{\hat{\cD}}$ and two $k$-consistent functions $h_{*}^1$, $h_{*}^2$ with common $R$-transcript $\alpha'$ after $k$-th interaction,
$$\Pr[h_*=h_{*}^1 \mid \alpha_2[k]=\alpha']=\Pr[h_*=h_{*}^2 \mid \alpha_2[k]=\alpha'],$$
where the randomness is over Game \ref{Game:construction-clean}.
\end{lemma}
\begin{proof}
  Observe that, for a fixed distinguisher, the transcript $\alpha_2[k]$
  is a deterministic function of $h$. As $h$ is uniformly
  distributed, the distribution obtained by conditioning on the value
  of any fixed function of $h$ is also uniform.
\end{proof}
  

\subsubsection{Unpredictability}

In this section we bound the probability of $\pred$.

For $0 \leq k<q_{\hat{\cD}}$, we define the event $\pred[k]$ to be: in Game \ref{Game:construction-clean}.2, $\pred$ occurs before the end of $k$-th query and answer. In the rest of Section 4.4, we denote by $T[k]$ the data in $T_H$ and $T_F$ in Game \ref{Game:construction-clean}.2 after $k$-th query and answer. 

\begin{corollary}\label{Coro:unpre}
For any $0<k<q_{\hat{\cD}}$, $\Pr[\pred[k] \wedge \neg \pred[k-1]] = O(\ell^2(k-1)q_{\tilde{H}}2^{-3n}+\sqrt{\ell}\epsilon^{1/16}).$
\end{corollary}
\begin{proof} We consider the event $\Pr[\pred[k] \wedge \neg \pred[k-1]]$.
\begin{align*}
  & \mkern20mu \Pr[\pred[k] \wedge \neg \pred[k-1]] \\
  & \leq \Pr[\pred[k] \wedge \neg \pred[k-1] \mid \text{$\alpha_2[k-1]$ is regular}]+\Pr[\text{$\alpha_2[k-1]$ is not regular}] \\
  & \leq \Pr[\pred[k] \wedge \neg \pred[k-1] \text{ and $(R, h_*)$ is normal} \mid \text{$\alpha_2[k-1]$ is regular}]\\
  & \mkern20mu +\Pr[\text{$(R, h_*)$ is not normal} \mid \text{$\alpha_2[k-1]$ is regular}]+\sqrt{\ell}\epsilon^{1/16}\\
  & \leq \Pr\left[\parbox{65mm}{the $k$-th query is a constellation $y$ that contains an invisible term and $\tilde{g}_R(y) \in T[k-1]$}\;\middle|\; \text{$\alpha_2[k-1]$ is regular}\right] +2\sqrt{\ell}\epsilon^{1/16}\\
  & \leq \sum_{i \in [\ell]} \Pr\left[\parbox{7cm}{the $k$-th query is a constellation $y$, $\tilde{g}_R(y) \in T[k-1]$ and the $i$th term  $(y,h_i)$ is invisible}\;\middle|\; \text{$\alpha_2[k-1]$ is regular}\right] +2\sqrt{\ell}\epsilon^{1/16}
\end{align*}
Now consider any fixed $\alpha_2[k-1]$ in conjunction with a fixed
index $i$. Recall that fixing $\alpha_2[k-1]$ determines $R$ and hence
the constellation determined by the $k$th query; thus $i$ uniquely
determines a term $(y,h_i)$ of the constellation $y$ queried by $\cD$
at the $k$th step. Consider further a ``partial assignment'' $p_*$ to
the functions $h_*$ that provides values at all points in the domain
except for $(y,h_i)$. We say that $p_*$ has an \emph{invisible
  completion} (w.r.t.\ $\alpha_2[k-1]$) if there is some assignment to
$(y,h_i)$ so that this term is invisible (for $p_*$ and the $R$ given
by $\alpha[k-1]$). With this language, we may further upper bound the
expression just above by
\begin{equation}
  \leq \sum_{i \in [\ell]} \Pr\left[\parbox{4cm}{ $\tilde{g}_R(y) \in T[k-1]$ and the $i$th term  $(y,h_i)$ is invisible}\;\middle|\; \parbox{7cm}{ $\alpha_2[k-1]$ is regular, the $k$-th query is a constellation $y$ with $i$th term $(y,h_i)$, the partial assignment $p_*$ excluding $(y,h_i)$ has an invisible completion}\right] +2\sqrt{\ell}\epsilon^{1/16}\,. \label{eq:pick-up}
\end{equation}
(Here, the partial assignment $p_*$ is simply the function
$h_*$ with the value at
$(y,h_i)$ excluded.) Observe that for any fixed triple
$(\alpha_2[k-1], i,
p_*)$ that may arise in the conditioning above, there are exactly
$2^{3n}$ completions of $p_*$ to a full function
$h_*$. Recalling the definition of invisible, it follows that at least
a
$1-\epsilon^{1/4}$ fraction of all completions are consistent with
$\alpha_2[k-1]$ and, furthermore, maintain the invisibility of
$(y,h_i)$. Note, additionally, that at most
$|T[k-1]|$ of these completions can have the property that
$\tilde{g}_R(y) \in
T[k-1]$ and, additionally, that all such completions have the same
conditional probability (Lemma~\ref{lemma:unpre}). We conclude that
the probability of~\eqref{eq:pick-up} is no more than
\begin{equation*}
    \leq \frac{\ell}{1- \epsilon^{1/4}}
    \frac{|T[k-1]|}{2^{3n}}+2\sqrt{\ell}\epsilon^{1/16} 
     = O(\ell^2(k-1)q_{\tilde{H}}2^{-3n}+\sqrt{\ell}\epsilon^{1/16})\,. \qedhere
\end{equation*}
\end{proof}

\begin{theorem}\label{thm:unpred.}
  $\Pr[\pred] = O(\ell^2q_{\hat{\cD}}^2q_{\tilde{H}}2^{-3n} +
  \sqrt{\ell}q_{\hat{\cD}}\epsilon^{1/16})$.
\end{theorem}

Theorem \ref{thm:unpred.} follows from the union bound to
Corollary~\ref{Coro:unpre} over $k$.

\subsubsection{Subversion freedom}

Now we bound the probability of the event $\subv$. 

\begin{definition}[Silent term]
  Given $(R,h_*)$, for any $x \in \{0,1\}^n$ and $i \in [\ell]$, we say $(x,h_i)$ is \emph{silent} if $(x,h_i)$ is invisible and is queried by fewer than $2^{2.5n}q_{\tilde{H}}$ terms of $\tilde{h}_0$.
\end{definition}

\begin{lemma}
For any $0<k<q_{\hat{\cD}}$, assuming that $\ell > n+4$,
\[
  \Pr[\text{there exists $x$ such that the constellation $x$ has no
    silent term} \mid \text{$\alpha_2[k]$ is regular}]=
  O\left(\frac{\ell^4}{2^n}+\sqrt{\ell} \epsilon^{1/16}\right)\,.
\]
\end{lemma}

\begin{proof}
  Using Markov's inequality, for each $x \in \{0,1\}^n$ and
  $i \in [\ell]$,
  \[
    \Pr_{R, h_*}[\text{$h_{i}(x \oplus r_i)$ is queried by $\tilde{h}_0(y)$ for more than $2^{2.5n}q_{\tilde{H}}$ many $y$'s}]<\frac{1}{\sqrt{2^n}}.
  \]
  Let us call such a $(x \oplus r_i, h_i)$ \emph{thick} if it is
  queried by so many $\tilde{h}_0(y)$ and \emph{thin} otherwise. (So
  that a term is silent if it is invisible and thin.)  Note that for
  any fixed constellation $x$, the events that $(x \oplus r_i, h_i)$
  are thin are independent (as they are determined by different
  $r_i$). Since $\ell > n+4$,
\begin{align*}
  & \Pr[\text{there exists $x$ such that the constellation $x$ has no silent term} \mid \text{$\alpha_2[k]$ is regular}]\\
  \leq{} &  \Pr_{R, h_*}\left[\parbox{8cm}{there exists $x$ such that the constellation $x$ has no silent term and has more than $\ell-n$ invisible terms} \;\middle|\; \text{$\alpha_2[k]$ is regular}\right]\\
  & \mkern20mu + \Pr_{R, h_*}\left[\text{there exists a constellation $x$ that has more than $n$ visible terms}\mid \text{$\alpha_2[k]$ is regular}\right]\\
  \leq{} &  \Pr_{R, h_*} \left[\parbox{8cm}{there exists $x$ such that the constellation $x$ has more than $\ell-n$ invisible terms and all these terms are thick}\right]\Big/\Pr[\text{$\alpha_2[k]$ is regular}] + O\left(\sqrt{\ell} \epsilon^{1/16}\right)\\
  \leq{} &  \Pr_{R, h_*} \left[\parbox{8cm}{there exists $x$ such that the constellation $x$ has more than $\ell-n$ thick terms}\right]\Big/\Pr[\text{$\alpha_2[k]$ is regular}] + O\left(\sqrt{\ell} \epsilon^{1/16}\right)\\
  \leq{} & 2^n \binom{\ell}{4}  \left(\frac{1}{\sqrt{2^n}}\right)^{4} \cdot 2 + O(\sqrt{\ell} \epsilon^{1/16}) = O(\ell^4/2^n+\sqrt{\ell} \epsilon^{1/16}).\qedhere
\end{align*}
\end{proof}

For $0 \leq k<q_{\hat{\cD}}$, we define the event $\subv[k]$ to be: in
Game \ref{Game:construction-clean}.2, $\subv$ occurs before the end of
$k$-th query and answer.

\begin{corollary}\label{Coro:subv}
For any $0<k<q_{\hat{\cD}}$, $\Pr[\subv[k] \wedge \neg \subv[k-1]] = O(\ell q_{\tilde{H}}/\sqrt{2^n}+\sqrt{\ell}\epsilon^{1/16}).$
\end{corollary}
\begin{proof} We consider the event $\subv[k] \wedge \neg \subv[k-1]$.
\begin{align*}
  & \mkern20mu \Pr[\subv[k] \wedge \neg \subv[k-1]]\\
  & \leq \Pr[\subv[k] \wedge \neg \subv[k-1] \mid \text{$\alpha_2[k-1]$ is regular}]+\Pr[\text{$\alpha_2[k-1]$ is not regular}]\\
  & \leq \Pr[\text{$\subv[k] \wedge \neg \subv[k-1]$ and every constellation has a silent term} \mid \text{$\alpha_2[k-1]$ is regular}]\\
  & \mkern20mu +\Pr[\text{not all constellations have a silent term} \mid \text{$\alpha_2[k-1]$ is regular}]+\sqrt{\ell}\epsilon^{1/16}\\
  & \leq \sum_{i \in [\ell]}\Pr\left[\parbox{7cm}{the $i$th term $(x,h_i)$ of the constellation $x$ queried at step $k$ is silent and $h_0$ is subverted at $\tilde{g}_R(x)$} \;\middle|\; \text{$\alpha_2[k-1]$ is regular}\right] +\ell^4/2^n+\sqrt{\ell}\epsilon^{1/16}
\end{align*}    
Now consider any fixed $\alpha_2[k-1]$ in conjunction with a fixed
index $i$. Recall that fixing $\alpha_2[k-1]$ determines $R$ and hence
the constellation determined by the $k$th query; thus $i$ uniquely
determines a term $(y,h_i)$ of the constellation $y$ queried by $\cD$
at the $k$th step. Consider further a ``partial assignment'' $p_*$ to
the functions $h_*$ that provides values at all points in the domain
except for $(y,h_i)$. We say that $p_*$ has a \emph{silent
  completion} (w.r.t.\ $\alpha_2[k-1]$) if there is some assignment to
$(y,h_i)$ so that this term is silent (for $p_*$ and the $R$ given
by $\alpha[k-1]$). With this language, we may further upper bound the
expression just above by
\begin{align*}
  & \leq \sum_{i \in [\ell]}\Pr\left[\parbox{3cm}{$(x,h_i)$ is silent and $h_0$ is subverted at $\tilde{g}_R(x)$} \;\middle|\; \parbox{7cm}{$\alpha_2[k-1]$ is regular, the partial assignment excluding the $i$th term $(x,h_i)$ of the constellation $x$ queried at step $k$ has a silent completion}\right]+\ell^4/2^n+\sqrt{\ell}\epsilon^{1/16}\\
  & \leq \sum_{i \in [\ell]}\Pr\left[\parbox{45mm}{$h_0$ is subverted at $\tilde{g}_R(x)$ and $\tilde{h}_0(\tilde{g}_R(x))$ does not query $(x,h_i)$} \;\middle|\; \parbox{7cm}{$\alpha_2[k-1]$ is regular, the partial assignment excluding the $i$th term $(x,h_i)$ of the constellation $x$ queried at step $k$ has a silent completion}\right]\\
  & \mkern20mu + \sum_{i \in [\ell]}\Pr\left[\parbox{4cm}{$h_0$ is subverted at $\tilde{g}_R(x)$ and $\tilde{h}_0(\tilde{g}_R(x))$ queries $(x,h_i)$} \;\middle|\; \parbox{7cm}{$\alpha_2[k-1]$ is regular, the partial assignment excluding the $i$th term $(x,h_i)$ of the constellation $x$ queried at step $k$ has a silent completion}\right]+\ell^4/2^n +\sqrt{\ell}\epsilon^{1/16}.
\end{align*}
(Here, the partial assignment $p_*$ is simply the function $h_*$ with
the value at $(y,h_i)$ excluded.) Observe that for any fixed triple
$(\alpha_2[k-1], i, p_*)$ that may arise in the conditioning above,
there are exactly $2^{3n}$ completions of $p_*$ to a full function
$h_*$. Recalling the definition of silent, it follows that at least a
$1-\epsilon^{1/4}$ fraction of all completions are consistent with
$\alpha_2[k-1]$ and, furthermore, maintain the silence of
$(y,h_i)$. Note, additionally, that at most $2^{3n} \epsilon$ of these
completions can have the property that $h_0$ is subverted at
$\tilde{g}_R(x)$ and $\tilde{h}_0(\tilde{g}_R(x))$ does not query
$(x,h_i)$, while at most $2^{2.5n}q_{\tilde{H}}$ of these completions
can have the property that $h_0$ is subverted at $\tilde{g}_R(x)$ and
$\tilde{h}_0(\tilde{g}_R(x))$ queries $(x,h_i)$. Moreover, all such
completions have the same conditional probability
(Lemma~\ref{lemma:unpre}). We conclude that the probability above is
no more than
\begin{equation*}
    \leq \ell\left(\frac{2^{3n}\epsilon}{2^{3n}}+ \frac{2^{2.5n}q_{\tilde{H}}}{2^{3n}}\right)/(1-\epsilon^{1/4})+\ell^4/2^n +\sqrt{\ell}\epsilon^{1/16} = O(\ell q_{\tilde{H}}/\sqrt{2^n}+\sqrt{\ell}\epsilon^{1/16}).\qedhere
\end{equation*}
\end{proof}

\begin{theorem}\label{thm:subv}
  $\Pr[\subv] = O(\ell q_{\hat{D}}q_{\tilde{H}}/\sqrt{2^n}+\sqrt{\ell}q_{\hat{D}}\epsilon^{1/16}).$
\end{theorem}

Theorem \ref{thm:subv} can be proved by applying the union bound to Corollary \ref{Coro:subv} over $k$.

\subsection{Controlling the self-referential probability}

Now we proceed to prove that the crisis event $\selfref$ is negligible. We prove this result by establishing a stronger theorem that says the entire data table in Game~\ref{Game:construction-clean}.1 has a negligible chance to have a self-referential constellation.

\begin{theorem}\label{thm:selfref} In
  Game~\ref{Game:construction-clean}.1, for any distinguisher $\cD$
  \[
    \Pr_{h, R}[\selfref] = O(q_{\tilde{H}} \ell 2^{-2n}).
  \]
\end{theorem}

We state a standard Chernoff bound; see~\cite{MotwaniR95}, or
\cite{Vu:Chernoff} for this particular formulation.
\begin{lemma}(Chernoff Bound)\label{lemma:chernoff}
Let $X_1$,...,$X_n$ be discrete, independent random variable such that $E[X_i]$=0 and $|X_i| \leq 1$ for all i. Let $X=\sum^n_{i=1} X_i$ and $\sigma^2$ be the variance of X. Then $\Pr[|X| \geq \lambda\sigma]\leq  2e^{-\lambda^2/4}$
for any $0\leq\lambda\leq 2\sigma$.
\end{lemma}

In the next two lemmas, we introduce the Fourier transformation to analyze the total variation distance between the distribution of
$\tilde{g}_R(x)$ and the uniform distribution. 

\begin{definition}[Character and Dual group]
  Suppose $G$ is a finite abelian group. A character $\chi$ of $G$ is
  a group homomorphism from $G$ into the multiplicative group $T$ of
  complex numbers of norm 1. We define the dual group $\hat{G}$ to be
  the set of all characters of $G$; these make an abelian group under
  the pointwise product operation
  $[\chi \sigma](x) = \chi(x) \cdot \sigma(x)$. We let
  $\mathbf{1}: G \rightarrow T$ denote the trivial character
  $\mathbf{1}: g \mapsto 1$. 
\end{definition}

It is a fact that $|G| = |\hat{G}|$ for any finite abelian group $G$
and, furthermore, that the functions of $\hat{G}$ are linearly
independent: that is, if $\sum_\chi a_\chi \chi$ is the zero function
for some collection of coefficients $a_\chi \in \CC$, then all
$a_\chi = 0$.

\begin{definition}[Discrete Fourier transform]
  Suppose $G$ is a finite abelian group. Let $L^2(G)$ be the set of
  all complex-valued functions on $G$. The discrete Fourier transform
  on $f \in L^2(G)$ is defined by
  $$\hat{f}(\chi)=\sum_{a \in G}f(a)\overline{\chi(a)}=\langle f,\chi \rangle, \  \text{for}\ \chi \in \hat{G}.$$
  Here $\langle \,, \rangle$ denotes the inner product of
  complex-valued functions on $G$. We remark that if
  $f: G \rightarrow \RR$ is a probability distribution then
  $\hat{f}(\mathbf{1}) = 1$ and $|\hat{f}(\chi)| \leq 1$ for all
  $\chi$
\end{definition}

\begin{lemma}
  Define convolution by
\[
  (f \ast g)(x)=\sum_{y \in G}f(y)g(x-y),\ \text{for}\ x \in G.
\]
Then
\[
  \widehat{f \ast g}(\chi)=\hat{f}(\chi) \cdot \hat{g}(\chi),\ \text{for all}\ \chi \in \hat{G}\,.
\]
\end{lemma}

\begin{lemma}\label{lemma:plancherel}
Let $Q$ be a probability distribution and $U$ be the uniform distribution on a finite abelian group $G$. Then, 
\begin{align*}
  \|Q-U\|_{\tv} \leq \frac{1}{2} \left (\sum\limits_{\chi \in \hat{G}, \chi \neq \mathbf{1}}\left|\hat{Q}(\chi)\right|^2 \right )^{1/2}\,.
\end{align*}
\end{lemma}

In our application of the Fourier transform, $G$ is the group
$(\mathbb{Z}/2)^{3n}$. For any fixed $i \in [\ell]$ and
$x \in \{0,1\}^n$, we denote the distribution of $h_i(x \oplus r_i)$,
$\tilde{h}_i(x \oplus r_i)$ and $\tilde{g}_R(x)$ by $p_i^x$,
${p'}_i^{x}$ and $P^x$, respectively. These are random variables
defined by selection of $h_*$ and $R$. For a fixed value of
$h_*$, we further define $p^x_{i,h_*}$, $p'^x_{i,h_*}$, and
$P^x_{h_*}$ to be the distributions of $h_i(x \oplus r_i)$,
$\tilde{h}_i(x \oplus r_i)$ and $\tilde{g}_R(x)$ respectively, over the randomness of $R$. In most cases below we omit $x$ in the notation when there is
no ambiguity.

\begin{lemma}\label{lemma:tvd}
For any $x,r \in \{0,1\}^n$ and any $t \in [\ell]$,
\begin{align*}
  \Pr_{h_*} \left [\|P_{h_*;r,t}^x-U\|_{\tv} \geq 2^{3n-1}(n2^{-n/2}+\epsilon)^{\ell/2} \right ] \leq 2^{-\Omega(n^2 \ell)},
\end{align*}
where $P_{h_*;r,t}^x$ is the distribution of $\tilde{g}_R(x)$,
conditioned on $h_*$, over uniform $R$ with the constraint
$r_t=r$. $U$ is the uniform distribution on $G$ and $\epsilon$ here is
the (negligible) disagreement probability of~\eqref{eq:equal} above.
\end{lemma}
\begin{proof}
  Consider an arbitrary $i \in [\ell]$ and nontrivial
  $\chi \in \hat{G}$ (that is, $\chi \neq \mathbf{1}$).  Observe that
  for any $h_*$,
  \[
    \hat{p}_{i,h_*} = \sum_{v \in (\ZZ/2)^{3n}} \Pr_{r_i}[h_i(x \oplus r_i)
    = v] \overline{\chi(v)} = \frac{1}{2^n} \sum_{r_i} \chi(h_i(x \oplus r_i))\,,
  \]
  as $\chi()$ is always real. When $h_i$ is selected uniformly, the
  $\chi(h_i(x \oplus r_i))$ are i.i.d.\ random variables taking values
  in $\{\pm 1\}$ and we may apply Lemma \ref{lemma:chernoff} with
  $\lambda=n$; noting that the variance is $\sigma^2 = 2^n$ this yields
\begin{align*}
    \Pr_{h_*}[2^n\hat{p}_{i,h_*}(\chi) \geq n2^{n/2}] \leq 2e^{-n^2/4}.
\end{align*}
For any fixed $h_i$ and $x$, by assumption
$\Pr_{r_i}[\tilde{h}_i(x \oplus r_i) \neq h_i(x \oplus r_i)] \leq \epsilon$
and hence
$|\hat{p}_{i,h_*}(\chi) - \hat{p}'_{i,h_*}(\chi)| \leq \epsilon$ for
every $\chi$. In light of this, for any $t$ and $\chi$,
\begin{align*}
  \Pr_{h_*}[|\hat{P}_{h_*;t,r}(\chi)| \geq (n2^{-n/2}+\epsilon)^{\ell/2}] & = \Pr_{h_*} \left [\mathop \prod  \limits ^{\ell}_{i=1,i \neq t} |\hat{p'}_{i, h_*}(\chi)| \geq (n2^{-n/2}+\epsilon)^{\ell/2} \right ]\\
                                                                            & \leq \Pr_{h_*}[\exists \text{ half of the coordinates $i \in [\ell]/\{r_t\}$ such that $|\hat{p'}_{i, h_*}(\chi)| \geq n2^{-n/2}+\epsilon$}]\\
                                                                            & \leq \Pr_{h_*}[\exists \text{ half of the coordinates $i \in [\ell]/\{r_t\}$ such that $|\hat{p}_{i, h_*}(\chi)| \geq n2^{-n/2}$}]\\
                                                                            & \leq \binom{\ell}{\ell/2} \cdot (2e^{-n^2/4})^{\ell/2} = (8e^{-n^2/4})^{\ell/2}\,.
\end{align*}
Thus,
\[
  \Pr_{h_*}[\exists \chi \in \hat{G} \setminus \{ \mathbf{1}\},\  |\hat{P}_{h_*;r,t}(\chi)| \geq (n2^{-n/2}+\epsilon)^{\ell/2}]
  \leq 2^{6n} (8e^{-n^2/4})^{\ell/2}
\]
and we conclude that
\begin{align*}
  \Pr_{h_*} \left [\|P_{h_*;r,t}^x-U\|_{\tv} \geq \frac{1}{2}2^{3n}(n2^{-n/2}+\epsilon)^{\ell/2} \right ]
    & = \Pr_{h_*} \left [\frac{1}{2}\sum\limits_{\chi \in \hat{G}, \chi \neq \mathbf{1}}|\hat{P}_{h_*;r,t}(\chi)| \geq 2^{3n-1}(n2^{-n/2}+\epsilon)^{\ell/2} \right ]\\
    & \leq \Pr_{h_*}[\exists \chi \in \hat{G}, \chi \neq \mathbf{1},\  |\hat{P}_{h_*;r,t}(\chi)| \geq (n2^{-n/2}+\epsilon)^{\ell/2}]\\
    & \leq 2^{6n}(8e^{-n^2/4})^{\ell/2} = 2^{-\Omega(n^2 \ell)},
\end{align*}
where the first equality is Lemma \ref{lemma:plancherel}.
\end{proof}

Now we are ready to show that, with overwhelmingly probability, there is no self-referential constellation in the entire table.

\begin{lemma}\label{lemma:no-con}
  In Game~\ref{Game:construction-clean}.1, for some $R$, $h_*$ and
  $x \in \{0,1\}^n$, we say constellation $x$ is self-referential if
  there exists $t \in [\ell]$ such that $h_0(\tilde{g}_R(x))$ is
  queried by $\tilde{h}_t(x \oplus r_t)$.  Then, assuming that
  $\ell > n$, 
  \[
    \Pr_{R, h_*}[\text{there exists a self-referential constellation}]=O(q_{\tilde{H}} \ell 2^{-2n}).
    \]
\end{lemma}
\begin{proof}
  For any $x,r \in \{0,1\}^n$ and any $t \in [\ell]$, we say $h_*$ is
  $(x,r,t)$-\emph{good} if
  $\|P_{h_*;r,t}^x-U\|_{\tv} \leq
  2^{3n-1}(n2^{-n/2}+\epsilon)^{\ell/2}$. Otherwise, we say $h_*$ is
  $(x,r,t)$-\emph{bad}. Let $P_{r,t}^x$ be the distribution of
  $\tilde{g}_R(x)$ over the randomness of $h_*$ and $R$ with the
  constraint $r_t=r$. Then,
\begin{align*}
  & \mkern20mu \Pr_{R, h_*}[\exists x\in\{0,1\}^n, \text{the constellation of $x$ is self-referential}]\\ & \leq  \sum_{x \in \{0,1\}^n} \sum_{t=1}^{\ell}\sum_{r \in \{0,1\}^n}\Pr[r_t=r]\Pr_{R/r_t,h_*}[\text{$h_0(\tilde{g}_R(x))$ is queried by $\tilde{h}_t(x \oplus r_t)$ with $r_t=r$}]\\
  & \leq \ell 2^{n} \max_{x,r,t} \Bigl\{\Pr_{R/r_t,h_*}[\text{$h_0(\tilde{g}_R(x))$ is queried by $\tilde{h}_t(x \oplus r)$} \mid \text{$h_*$ is $(x,r,t)$-good}]\cdot \Pr_{h_*}[\text{$h_*$ is $(x,r,t)$-good}]\\
  & {}\phantom{\leq \ell 2^{n} \max_{x,r,t} \Bigl\{ {}} +\Pr_{h_*}[\text{$h_*$ is $(x,r,t)$-bad}] \Bigr\} \\
  & \leq \ell 2^{n} \left\{(2^{3n-1}(n2^{-n/2}+\epsilon)^{\ell/2}+2^{-3n}  \cdot q_{\tilde{H}})+O(2^{-\Omega(n^2\ell)})\right\}\\
  & =O(q_{\tilde{H}} \ell 2^{-2n})\,,
\end{align*}
so long as $\ell > n$. (Note that for sufficiently large $n$, we have
$(n2^{-n/2}+\epsilon)^{\ell/2} =O(n^{\ell/2}2^{-\Theta(n\ell)}+\epsilon^{\Theta(\ell)})$.) In the last inequality we use the fact
that if $\| D_1 - D_2\| \leq \delta$ for two distributions $D_1$ and
$D_2$ then, for any event $E$,
$\Pr_{D_1}[E] \leq \Pr_{D_2}[E] + \delta$.
\end{proof}

Theorem~\ref{thm:selfref} follows immediately from Lemma~\ref{lemma:no-con}.  

\subsection{Crooked indifferentiability in the full model}

Now we show the simulator $\cS$ achieving abbreviated crooked
indifferentiability can be lifted to a simulator that achieves full
indifferentiability (Definition~\ref{def:abbrev-indiff-crooked}).

\begin{theorem}\label{Abbr. to full indiff}
If the construction in Section \ref{sec:construction} is $(n',n,q_{\cD},q_{\tilde{H}},r,\epsilon')$-Abbreviated-$H$-crooked-indifferentiable from a random oracle $F$, it is $(n',n,q_{\cD},q_{\tilde{H}},r,\epsilon'+O(q_{\cD}^2q_{\tilde{H}} \ell 2^{-3n}))$-$H$-crooked-indifferentiable from $F$.
\end{theorem}

\begin{proof}
  Consider the following simulator $\cS_{F}$ built on $\cS$:
\begin{enumerate}
    \item Draw two data sets, $\DS_1$ and $\DS_2$: $\DS_1$ contains uniformly selected values for each $F(x)$
    and $h_i(x)$ for all $0 < i \leq \ell$ and all $x \in \{0,1\}^n$;
    $\DS_2$ contains uniformly selected values $h_0(y)$ for all
    $y \in \{0,1\}^{3n}$.
    \item In the first phase, $\cS_{F}$ answers $h_i(x)$ ($0 < i \leq \ell$) and $h_0(y)$ queries according to $\DS_1$ and $\DS_2$, respectively. 
    \item The second phase, after which $\cS_F$ receives $R$, is
      divided into two sub-phases.
    \begin{itemize}
    \item First, $\cS_F$ simulates $\cS$ in Game \ref{Game 4} with
      data sets $\DS_1$ and $\DS_2$ generated above. It then plays the
      role of the distinguisher, and asks $\cS$ all the questions that
      were actually asked by the the distinguisher in the first
      phase. $\cS_F$ \emph{aborts} the game if, in this sub-phase,
      there are two (simulated) queries $(x,h_0)$ and $(y,h_i)$
      ($0 < i \leq \ell$) such that $x=\tilde{g}_R(y \oplus r_i)$.
    \item Second, $\cS_{F}$ simulates $\cS$ and answers the
      second-phase questions from the distinguisher.
    \end{itemize}
\end{enumerate}
For an arbitrary full model distinguisher $\cD_F$, we construct the an
abbreviated model distinguisher $\cD$ as follows. The proof will show
that, with high probability, the execution that takes place between
$\cD$ and $\cS$ can be ``lifted'' to an associated execution between
$\cD_F$ and $\cS_F$.
\begin{enumerate}
\item Prior to the game, $\cD$ must publish a subversion algorithm
  $\tilde{H}$. This program is constructed as follows. To decide how
  to subvert a certain term $h_i(x)$, $\tilde{H}$ first simulates the
  first phase of $\cD_F$; all queries made by this simulation are
  asked as regular queries by $\tilde{H}$ and, at the conclusion, this
  first phase of $\cD_F$ produces, as output, a subversion algorithm
  $\tilde{H}_F$. $\tilde{H}$ then simulates the algorithm
  $\tilde{H}_F$ on the term $h_i(x)$.
\item In the game, $\cD$ simulates the queries of $\cD_{F}$ in
  $\cD_{F}$'s first phase. After that, $\cD$ continues to simulate
  $\cD_{F}$ in the second phase. (Note that at the point in
  $\cD_{F}$'s game where it produces the subversion algorithm
  $\tilde{H}$, this is simply ignored by $\cD$.)
\end{enumerate}

Now we are ready to prove $\cS_{F}$ is secure against the arbitrarily
chosen distinguisher $\cD_F$. We organize the proof around four different transcripts:
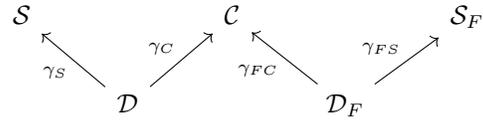
\begin{center}
\begin{tabular}{|c|l|}
  \hline
  $\gamma_{FC}$ & transcript of $\cC$ interacting with $\cD_{F}$\\ \hline
  $\gamma_C$ & transcript of $\cC$ interacting with $\cD$\\ \hline
  $\gamma_{FS}$ & transcript of $\cS_{F}$ interacting with $\cD_F$\\ \hline
  $\gamma_{S}$  & transcript of $\cS$ interacting with $\cD$\\ \hline
\end{tabular}
\qquad
\begin{tikzcd}
  \cS  &     & \cC &       & \cS_F\\
  & \cD \arrow[ul,"\gamma_S"] \arrow[ur,"\gamma_C"] &     & \cD_F \arrow[ul,"\gamma_{FC}"] \arrow[ur,"\gamma_{FS}"]& 
\end{tikzcd}
\end{center}
(Here $\cC$ denotes the construction, as usual.)
  %
Since
\[
  \|\gamma_{FC}-\gamma_{FS}\|_{\tv} \leq
  \|\gamma_{FC}-\gamma_C\|_{\tv}+\|\gamma_C-\gamma_S\|_{\tv}+\|\gamma_S-\gamma_{FS}\|_{\tv}\,,
\]
it is sufficient to prove the three terms in the right-hand side of
the inequality are all negligible.
\begin{itemize}
\item[-]  $\|\gamma_{FC}-\gamma_C\|_{\tv}=0$.
  This is obvious by observing that $\gamma_{FC}=\gamma_C$ when the underlying values of $H$ are the same.
\item[-]  $\|\gamma_C-\gamma_S\|_{\tv}= \epsilon'$.
  This is true because $\cS$ achieves abbreviated crooked indifferentiability.
\item[-] $\|\gamma_S-\gamma_{FS}\|_{\tv}= O(q_{\cD}^2q_{\tilde{H}} \ell 2^{-3n})$.  To prove this
  statement, we suppose both the full model game and the abbreviated
  model game select all randomness \emph{a priori} (as in the
  descriptions above). Suppose the game between $\cS_F$ and $\cD_F$
  and the game between $\cS$ and $\cD$ share the same data sets
  $\DS_1$, $\DS_2$ and $R$. Notice that the two games have same
  transcripts unless, $\cS_F$ aborts the game in the first sub-phase
  of the second phase. We denote this bad event by \textbf{Conflict},
  which by the following lemma \ref{lemma: many}, is negligible. \qedhere
 \end{itemize}
\end{proof}

N.b. While the description of the simulator above calls for all
randomness to be generated in advance, it is easy to see that the
simulator can in fact be carried out lazily with tables.

To explore the failure event above, we define the following game,
named \textbf{Exp-Many}, played by an unbounded adversary
$\mathcal{M}$.

\begin{mdframed}
\begin{center}
  \textsc{Exp-Many}
\end{center}
\begin{enumerate}
    \item Select two data sets $DS_1$ and $DS_2$: $\DS_1$ contains uniformly selected values for each $F(x)$
    and $h_i(x)$ for all $0 < i \leq \ell$ and all $x \in \{0,1\}^n$;
    $\DS_2$ contains uniformly selected values $h_0(y)$ for all
    $y \in \{0,1\}^{3n}$.
    \item Given $DS_1$ and $DS_2$, $\mathcal{M}$ publishes a
      subversion algorithm $\tilde{H}$ and a sequence of
      $p(n)(<q_{\cD})$ terms in the form of
      $(y,h_i)$ ($0< i \leq \ell$) or $(x,h_0)$.
    \item $R$ is selected uniformly.
    \item Implement Game~\ref{Game 4} with $DS_1$, $DS_2$, $R$, $\tilde{H}$ and the queries prepared above.
    \item Output 1 if, in Game~\ref{Game 4} of the last step, there
      exist two queries $(x,h_0)$ and $(y,h_i)$ ($i>0$) such that
      $x=\tilde{g}_R(y \oplus r_i)$. Otherwise, output 0.
\end{enumerate}
\end{mdframed}

\begin{lemma}\label{lemma: many}
For any distinguisher $\mathcal{M}$, $\Pr[\text{\textbf{Exp-Many} outputs 1}]=O(q_{\cD}^2q_{\tilde{H}} \ell 2^{-3n})$.
\end{lemma}

A quick thought reveals that Lemma \ref{lemma: many} implies the negligibility of \textbf{Conflict}. Suppose there exists a distinguisher $\mathcal{M}$ in the full model such that \textbf{Conflict} is non-negligible. Consider the following distinguisher $\mathcal{M'}$ in \textbf{Exp-Many}: In Step 2 of \textbf{Exp-Many}, simulate $\mathcal{M}$ and publish the queries $\mathcal{M}$ publishes. Obviously, by definition of \textbf{Exp-Many}, the probability that \textbf{Exp-Many} outputs 1 when against $\mathcal{M'}$ is non-negligible.

Now we proceed to prove Lemma \ref{lemma: many}. We introduce several concepts that are useful in the proof. Given data sets $DS_1$, $DS_2$ and a subversion algorithm $\tilde{H}$, for any term $(x,h_i)$ with $0< i \leq \ell$ and $x \in \{0,1\}^n$, we define its \emph{ideal subversion} $I(x,i)$ to be the subverted value of $h_i(x)$ via the subversion algorithm $\tilde{H}$, using $h_0$ values in $DS_2$ and $h_i$ ($i>0$) values in $DS_1$. The trace of $(x,h_i)$ is defined to be the set
\newcommand{\Tr}{\mathsf{Tr}}
\[
  \Tr(x,i)=\{y \in \{0,1\}^{3n} \mid \text{$(y,h_0)$ is queried in the evaluation of $I(x,i)$}\}\,.
\]

Given $DS_1$, $DS_2$, $R$ and $\tilde{H}$, for any constellation $x$, its ideal output is
\[
  I(x)=\bigoplus_{i = 1}^{\ell} I(x \oplus r_i,i)\,,
\quad\text{and its \emph{trace} is defined by}\quad
  \Tr(x)=\bigcup_{i=1}^{\ell}\Tr(x \oplus r_i,i)\,.
\]

In \textbf{Exp-Many}, we define the event 
\[
\text{Crossref} = \left\{\text{In step 4, a constellation $x$ is queried such that $I(x) \neq \tilde{g}_R(x)$} \right\}\,.
\]

\begin{lemma}\label{crossref}
For any distinguisher $\mathcal{M}$ in \textbf{Exp-Many}, $\Pr[\mathrm{Crossref}]=O(q_{\cD}^2q_{\tilde{H}} \ell 2^{-3n})$.
\end{lemma}

\begin{proof}
  Notice that, if \textbf{Crossref} occurs, there are two queries
  $(x,h_i)$ and $(y,h_j)$ ($0<i,j< \ell$) such that $I(x \oplus r_i) \in
  \Tr(y \oplus r_j)$. We define the event \textbf{Two-cross} to be
  \[
    \parbox{10cm}{In \textbf{Exp-Many}, among the $p(n)$ queries made by $\mathcal{M}$, there are two terms $(x,h_i)$ and
        $(y,h_j)$ ($0<i,j< \ell$) such that
        $I(x \oplus r_i) \in
  \Tr(y \oplus r_j)$.}
\]
It is sufficient to show $\Pr[\textbf{Two-cross}]= O(q_{\cD}^2q_{\tilde{H}} \ell 2^{-3n})$.

For any two terms $(x,h_i)$ and $(y,h_j)$ ($0<i,j< \ell$) in the
queries of $\mathcal{M}$, whether $I(x \oplus r_i) \in
  \Tr(y \oplus r_j)$ only depends on
$DS_1$, $DS_2$, $\tilde{H}$ and $R$ (i.e., it has nothing to do with
the queries chosen by $\mathcal{M}$).

For any $x,r \in \{0,1\}^n$ and any $t \in [\ell]$, we say $(DS_1,DS_2)$ is
  $(x,r,t)$-\emph{good} if
  $\|P_{h_*;r,t}^x-U\|_{\tv} \leq
  2^{3n-1}(n2^{-n/2}+\epsilon)^{\ell/2}$, where $P_{h_*;r,t}^x$ is the distribution of $I(x)$,
conditioned on $DS_1$ and $DS_2$, over uniform $R$ with the constraint
$r_t=r$, and $U$ is the uniform distribution. Using the same proof in Lemma \ref{lemma:tvd}, we can show for any $x \in \{0,1\}^n$ and $t \in [\ell]$,
\begin{align*}
  \Pr_{DS_1,DS_2,r}  [\text{$(DS_1,DS_2)$ is
  $(x,r,t)$-\emph{good}}] \geq 1-2^{-\Omega(n^2 \ell)}.
\end{align*}

By Markov's inequality and the union bound, with probability smaller than $\ell 2^{-\Omega(n^2 \ell)}$ in the choice of $(DS_1,DS_2)$, we have 
\begin{align*}
  \Pr_{r}  [\text{$(DS_1,DS_2)$ is not
  $(x,r,t)$-\emph{good}} \mid (DS_1,DS_2)] \leq 2^{-\Omega(n^2 \ell)},
\end{align*}
for any $x \in \{0,1\}^n$ and $t \in [\ell]$. We say $(DS_1,DS_2)$ is all-good if it satisfies the inequality above.

By a Fourier transform argument similar to that in
Lemma~\ref{lemma:no-con}, if $(DS_1,DS_2)$ is all-good in \textbf{Exp-Many}, then for any $\tilde{H}$, any $0<i,j\leq \ell$ and any $x,y \in \{0,1\}^n$, and $\ell>n$,
$$ \Pr_{R}[I(x \oplus r_i) \in \Tr(y \oplus r_j)] \leq \sum_{k=1}^{\ell} \Pr_{R}[I(x \oplus r_i) \in \Tr(y \oplus r_j,k)] = O(q_{\tilde{H}} \ell 2^{-3n})\,$$

Finally,
\begin{align*}
    & \mkern20mu \Pr[\textbf{Two-cross}]\\
    & < \Pr[\text{$(DS_1,DS_2)$ is all-good}]+\Pr[\textbf{Two-cross} \mid \text{$(DS_1,DS_2)$ is not all-good}]\\  
    & < 2^{-\Omega(n^2 \ell)}+\binom{p(n)}{2}O(q_{\tilde{H}} \ell 2^{-3n})\\
    & = O(q_{\cD}^2q_{\tilde{H}} \ell 2^{-3n}).
\end{align*}
\end{proof}

To complete the proof, we consider an even simpler game that focuses attention on a pair of queries.

\begin{mdframed}
\begin{center}
  \textsc{Exp-One}
\end{center}
\begin{enumerate}
    \item Select two data sets $DS_1$ and $DS_2$: $\DS_1$ contains uniformly selected values for each $F(x)$
    and $h_i(x)$ for all $0 < i \leq \ell$ and all $x \in \{0,1\}^n$;
    $\DS_2$ contains uniformly selected values $h_0(y)$ for all
    $y \in \{0,1\}^{3n}$.
    \item Given $DS_1$ and $DS_2$, $\mathcal{M'}$ publishes a subversion algorithm $\tilde{H}$ and a triple $((x,h_i),y)$ for some $0\leq i \leq \ell$, $x \in \{0,1\}^n$ and $y \in \{0,1\}^{3n}$.
    \item $R$ is selected uniformly.
    \item Implement Game~\ref{Game 4} with $DS_1$, $DS_2$, $R$, $\tilde{H}$ prepared above and the query $(x,h_i)$.
    \item Output 1 if in the last step of Game~\ref{Game 4},
      $\tilde{g}_R(x)=y$. Otherwise, output 0.
\end{enumerate}
\end{mdframed}

\begin{lemma}\label{lemma:many vs one}
  For any adversary $\mathcal{M}$ in \textbf{Exp-Many}, there exists an adversary $\mathcal{M'}$ in \textbf{Exp-One} such that 
\[
\Pr[\text{\textbf{Exp-One} outputs 1 against $\mathcal{M'}$}] > \frac{1}{p^2(n)} (\Pr[\text{\textbf{Exp-Many} outputs 1 against $\mathcal{M}$}]-\Pr[\text{Crossref}] )
\]
\end{lemma}
\begin{proof}
For an arbitrary adversary $\mathcal{M}$ in \textbf{Exp-Many}, consider the following adversary $\mathcal{M'}$ in \textbf{Exp-One}:
\begin{itemize}
    \item In step 2 of \textbf{Exp-One}
    \begin{enumerate}
    \item When given $DS_1$ and $DS_2$ in step 2 of \textbf{Exp-One}, $\mathcal{M'}$ simulates $\mathcal{M}$, gets a sequence of polynomial many terms in the form of $(x,h_i)$, and a subversion algorithm $\tilde{H}$. 
    \item Write the sequence as the union of two sets $A$ and $B$, where $A$ consists of elements of the form $(x,h_i)$ ($i>0$) and $B$ consists of elements of the form $(y,h_0)$. Uniformly select an element $(x,h_i)$ from $A$ and an element $(y,h_0)$ from $B$. Output $((x,h_i),y)$ and the subversion algorithm $\tilde{H}$.
    \end{enumerate}
  \item In step 4 of \textbf{Exp-One}\\
    $\mathcal{M'}$ play Game~\ref{Game 4} according to $DS_1$,
    $DS_2$, $R$ and $\tilde{H}$ determined in previous steps. In this
    game, $\mathcal{M'}$ queries $(x,h_i)$.
\end{itemize}

For \textbf{Exp-Many} with the adversary $\mathcal{M}$, if
\text{Crossref} does not occur, then for any term $(x,h_i)$ ($i>0$) in step 2, we have
$I(x \oplus r_i)=\tilde{g}_R(x \oplus r_i)$. If at the same time
\textbf{Exp-Many} outputs 1, there must exist $(y,h_0)$ and $(x,h_i)$
($i>0$) in step 2 such that
$y=\tilde{g}_R(x \oplus r_i)=I(x \oplus r_i)$.
\begin{align*}
    & \mkern20mu \Pr[\text{\textbf{Exp-One} outputs 1 against $\mathcal{M'}$}] \\
    & > \frac{1}{p^2(n)}\Pr[\text{There must exist $(y,h_0)$ and $(x,h_i)$ ($i>0$) in the sequence such that $y=\tilde{g}_R(x \oplus r_i)=I(x \oplus r_i)$}]\\
    & > \frac{1}{p^2(n)} (\Pr[\text{\textbf{Exp-Many} outputs 1}]-\Pr[\text{Crossref}]).
\end{align*}
\end{proof}

\begin{lemma}\label{lemma:one}
$\Pr[\text{\textbf{Exp-One} outputs 1}]=O(q_{\tilde{H}} \ell 2^{-3n}).$
\end{lemma}
\begin{proof}
It is sufficient to show $\tilde{g}_R(x \oplus r_i)$ is unpredictable for any $(x,h_i)$ ($i>0$) in $DS_1$. 

The unpredictability comes from the fact that $R$ is unknown when the adversary publishes the queries. For any $(x,h_i)$ ($i>0$) in $DS_1$, select and fix an arbitrary $r_i$. The distribution of $\tilde{g}_R(x \oplus r_i)$ is then same as that of the sum of $\tilde{h}_j(x \oplus r_i \oplus r_j)$ for $j \neq i$, which can be proved to be unpredictable by a Fourier transform similar to that in Lemma~\ref{lemma:no-con}.
\end{proof}

\begin{proof}[Proof of Lemma \ref{lemma: many}]
Lemma \ref{lemma: many} can be derived by simply summing up Lemma \ref{crossref}, Lemma \ref{lemma:many vs one} and Lemma \ref{lemma:one}.
\end{proof}

\begin{proof}[Proof of Theorem \ref{thm:ind}]
Theorem \ref{thm:ind} can be proved by taking the sum of inequalities in Theorem \ref{thm:unpred.}, Theorem \ref{thm:subv}, Theorem \ref{thm:selfref} and Theorem \ref{Abbr. to full indiff}.
\end{proof}

\ignore{

\begin{corollary}\label{cor:main} Consider a pair of algorithms $(H, Q)$ and a
  polynomial $p(n)$. Given a random string $z$ of length $p(n)$,
  $H^h(z; x)$ defines a subverted random oracle $\tilde{h}$, as above; the
  algorithm $Q^h(z;R)$, also given access to the randomness $R$
  defining $\tilde{h}$, generates a family of adaptive queries
  $q_1, \ldots, q_s$. Assume that for all $h$ and $z$,
  \begin{equation}\label{eq:test}
    \Pr_{x \in \{0,1\}^n}[\tilde{h}(x) \neq h(x)] = \negl(n)\,.
  \end{equation}
  Then, with high probability in $R$ and $z$ and conditioned on the
  $h(q_1), \ldots, h(q_s)$, for all $x$ outside the ``queried'' set $\{x \mid \text{$h_i(x \oplus r_i)$ was queried}\}$,
  the distribution of $\tilde{h}_R(x)$ has negligible
    statistical distance from uniform.
\end{corollary}
}



\section{Conclusions and Open Problems}
In this paper, we initiate the study of correcting subverted random oracles which are adversarially tampered and disagree with the original random oracle at a negligible fraction of inputs. We give a simple construction that can be proven indifferentiable from a random oracle. Our analysis involves, for a given output produced: identifying a good term of the construction which is both honestly evaluated and independent of other terms for its input; and developing a new machinery of rejection resampling lemma (to assure the existence of this term). Our work provides a general tool to transform a buggy implementation of random oracle into a well-behaved one and directly applies to the kleptographic setting.


There are many interesting problems worth further exploring. Here we only list a few. First, a better construction that may tolerate a larger fraction of errors in the subverted random oracle. Second, develop a parallel theory of self-correcting distributions. Third, consider correcting other subverted ideal objectives such as ideal cipher. Fourth, given our extended model of indifferentiability in the presence of crooked elements, consider stronger indifferentiability models (e.g., with a global random oracle, more robust replacement theorem with more subverted components). Fifth, we may consider other type of constructions such as the Feistel structure or sponge construction that maybe more efficient. Last but not least, can we build connections to (or even a unified theory of) error correction, self-correcting programs, randomness extraction  and our problem of correcting subverted ideal objects?

\section*{Acknowledgement}
Alexander Russell is supported in part by NSF grant CNS \#1801487, Qiang Tang was supported in part by NSF grant CNS \#1801492, and Hong-Sheng Zhou is supported in part by NSF grant CNS \#1801470.

\bibliographystyle{alpha}

\newcommand{\etalchar}[1]{$^{#1}$}

\end{document}